\documentclass[12pt,a4paper]{article}

\usepackage{hyperref}

\usepackage[top=2.5cm,bottom=2.5cm,left=2.2cm,right=2.2cm]{geometry}
\usepackage{latexsym,amsmath,amsfonts,amssymb,amsthm,amscd,mathrsfs}

\usepackage{xcolor}
  
\usepackage{stmaryrd}
\usepackage{MnSymbol}

\newcommand{\be}{\begin{equation}}
\newcommand{\ee}{\end{equation}}
\newcommand{\bea}{\begin{eqnarray}}
\newcommand{\eea}{\end{eqnarray}}
\newcommand{\ba}{\begin{aligned}}
\newcommand{\ea}{\end{aligned}}
\numberwithin{equation}{section}
\newcounter{thmcounter}
\numberwithin{thmcounter}{section}

\theoremstyle{definition}

\newtheorem{definition}[thmcounter]{Definition}
\newtheorem{remark}[thmcounter]{Remark}
\theoremstyle{plain}

\newtheorem{lemma}[thmcounter]{Lemma}
\newtheorem{proposition}[thmcounter]{Proposition}

\newenvironment{proofof}[2]{{\itshape \noindent Proof of #1 \ref{#2}.}}{\hfill\(\square\)}
\def\1{{\boldsymbol 1}}                     %
\def\0{{\boldsymbol 0}}                     %
\def\CC{{\mathbb C}}                    %
\def\RR{{\mathbb R}}                    %
\def\XX{{\mathbb X}}                    %

\def\ZZ{{\mathbb Z}}

\def\cC{{\mathcal C}}                       %
\def\g{{\mathfrak g}}
\def\h{{\mathfrak h}}
 \def\b{{\mathfrak b}}
 \def\k{{\mathfrak k}}
 \def\fd{{\mathfrak d}}
 \def\n{{\mathfrak n}}
\def\SL{{\rm SL}}                           %
\def\GL{{\rm GL}}                           %
\def\U{{\rm U}}                                     %
\def\tr{\mathrm{tr\,}}                        %
\def\ri{{\rm i}}                            %

\def\hermp{\mathrm{Herm}^+}

\def\sgn{\mathrm{sgn}}

\def\S{\Sigma}
\def\Lm{{\Lambda}}

\newcommand{\bpm}{\begin{pmatrix}}
\newcommand{\epm}{\end{pmatrix}}
\def\cinf{C^\infty}
\def\hol{{Hol}}
\def\half{{\textstyle{\frac12}}} 
\def\dt {\left.\frac{d}{dt}\right|_{t=0}}

\def\llangle {\langle\!\langle}  \def\rrangle {\rangle\!\rangle}
\def\darg{\ \cdot\ ,\ \cdot\ } \def\ddarg{\, \cdot\ ,\, \cdot\ }

\def\bk{\mathbf{k}} \def\bl{\mathbf{l}} \def\br{\mathbf{r}}

\begin{document}

\begin{center}
 {\Large\bf
The semi-direct product of Poisson $G$-spaces }
\end{center}

\medskip
\begin{center}
I.~Marshall
\\

\bigskip

Faculty of Mathematics, Higher School of Economics\\
 National Research University\\
 Usacheva 6, Moscow, Russia\\
 e-mail: imarshall@hse.ru

\end{center}

 \setcounter{tocdepth}{2}
 
 \medskip
 \begin{abstract}
The notion of semi-direct product of Poisson $G$-spaces is applied to illuminate examples arising in spin-extensions of Ruijsenaars models.
 \end{abstract}


{\linespread{0.7}\tableofcontents}



\section{Introduction}\label{intro}
The motivation for this work was to clean up the presentation in \cite{FFM} of a structure emerging from analysis of a class of spin extended models of Ruijsenaars type. As usual with Ruijsenaars models, this work deals with Poisson Lie groups, and the spin systems are described by Hamiltonian reduction applied to Poisson Lie group symmetries on appropriate spaces. In the work of \cite{FFM}, \cite{AO} the ``spin degrees of freedom'' are represented by several copies of a complex vector space $V$, on which there acts either the Poisson Lie group $\GL(V)$ 
or the unitary group. As was explained in \cite{FFM}, for the unitary case the Poisson structure on $V$ should be one discovered by Zakrzewski, and described in \cite{Z}. A transformation---referred to as the \textit{half-dressing transformation}---of the space consisting of several copies of $V$ was useful for implementing the reduction. Upon closer scrutiny, this transformation of $V^d$ to itself produces a nice Poisson structure. 
Following this insight, it was natural to compare with the formulae arising in \cite{AO}, whereupon it was found that there too, half-dressing is a useful device for describing the Poisson bracket on the space of spin variables.

Suppose that we have an ordered collection of vectors $v_1,v_2,\dots,v_d$ in $\CC^n$. To each $v_k$ we associate the positive definite hermitian matrix $\Phi_k=1+v_kv_k^\dag$,\footnote{The symbol $\dag$ denotes  `complex conjugate transpose'} which in turn can be uniquely represented in the form $\Phi_k = b_kb_k^\dag$, with $b_k$ an upper triangular matrix having real, strictly positive elements on the diagonal. Thus we obtain a map $(v_1,\dots,v_d)\mapsto(b_1(v_1),\dots,b_d(v_d))$, and we are interested in the product $b=b_1\dots b_d$. This matrix may be represented by the positive definite matrix $\Phi = bb^\dag$, and it is not hard to see that $\Phi$ is naturally expressed in the form
\[
\Phi = 1 + \tilde v_1\tilde v_1^\dag + \dots + \tilde v_d\tilde v_d^\dag
\]
and we call the vectors $(\tilde v_1,\dots,\tilde v_d)$ the half-dressing of the vectors $(v_1,\dots,v_d)$.
It is of interest to look at Poisson properties of the half-dressing map.

In the present article, the half-dressing formula is interpreted in the Poisson context via the notion of \textit{semidirect product of Poisson spaces}, which is available when a Poisson symmetry is present, 
by means of which the formulae in \cite{FFM} and in \cite{AO} arise as particular examples.

In order to make the article self-contained, Section \ref{PLGintro} provides background information on the subject of Poisson Lie groups.
The main idea is explained in Section \ref{semidirectsec}.
In the remaining sections details of the two examples are fleshed out: the real $\U(n)$ symmetry example corresponding to \cite{FFM}, and the holomorphic $\GL(n,\CC)$ example corresponding to \cite{AO}. Appendix \ref{zakreg} is a discussion of the relation to \cite{Z}, and Appendix \ref{sympstruct} deals with computation of the symplectic form corresponding to the Poisson structure taken from that source.

It was drawn to my attention that aspects of the current work can be found in other settings.
In consequence, and without dwelling on them too much, I will add the following remarks.

(a) What, on a whim, I opted to call `semi-direct product Poisson structure', belongs to the larger family given the name `\textit{mixed product Poisson structures}' by Lu and Mouquin in \cite{LM}, and coinciding with what Zakrzewski in \cite{Z} called `\textit{crossed product Poisson structures}' . The semi-direct Poisson structure appears to be a general case of an example in Section 7 of \cite{LM}. 

(b) The results described in the present article have connections with \textit{multiplicative quiver varieties}---for example a formula reminiscent of the momentum map defined in \eqref{holmomGH} can be found in \cite{GJS} (equation (5.1))---and they would be associated to the simplest nontrivial quiver $(\bullet\longrightarrow\bullet)$.

\subsection{Acknowledgements}
I would like to thank my coauthors, Maxime Fairon and Laszlo Feher, of the article from which the ideas presented here originated.
Thanks are also due for the unprecedentedly helpful contribution from the referee who endorsed the publication of this article in Journal of Geometry and Physics, and who not only corrected some of my foolish mistakes, but made many suggestions for improving the content. This rare generosity was greatly appreciated.



%
%

\section{Overview of Poisson Lie groups}\label{PLGintro}

This section is a summary of several properties of Poisson Lie groups. 
Proofs and further explanations can all be found in the original article of Semenov-Tian-Shansky \cite{STS1}. More recent accounts are \cite{STS2}, \cite{ES}.

\medskip
In general, $G$ will denote a Lie group (real or complex), and its Lie algebra will be written $Lie(G)=\g$. 

\subsection{Definition of a Poisson Lie group and the r-matrix}

\begin{definition}
A Poisson Lie Group $G$ is a Lie group together with a Poisson structure, with respect to which the multiplication law in $G$ defines a Poisson map from $G\times G$ to $G$. That is, if, for any $F,H\in Fun(G)$,
\be
\{F(\,\cdot\,g_2),H(\,\cdot\,g_2)\}(g_1) + \{F(g_1\,\cdot\,),H(g_1\,\cdot\,)\}(g_2) = \{F,H\}(g_1g_2)\qquad\forall g_1,g_2\in G,
\ee
with the obvious notation $F(\,\cdot\,g_2)(g) = F(gg_2)$ and $F(g_1\,\cdot\,)(g)=F(g_1g)$. Here `$Fun(G)$' stands for $\cinf(G)$ if $G$ is real, or $Hol(G)$ (holomorphic functions on $G$)  if $G$ is complex.
\end{definition}
The Poisson structure on a Poisson Lie group $G$ is necessarily zero at the identity element $e\in G$. Hence we may linearise  the Poisson bracket at the identity to obtain a Lie bracket on the dual space $\g^*$ of the Lie algebra $\g=Lie(G)$: denote this bracket $[\darg]_*$.

Often we make the convenient assumption that $\g$ have a non-degenerate, invariant inner-product $\langle\darg\rangle$, and use it to identify $\g$ with its dual $\g^*$. Depending on the context, the same notation is then used for pairings between dual vector spaces, but this will be obvious from the context, and should not cause confusion. 
\begin{definition}
Suppose that $\g$ has a non-degenerate, invariant inner-product $\langle\darg\rangle$. Then the canonical 3-form on $\g$ is
\be
\phi(X,Y,Z) = \langle X,[Y,Z]\rangle,
\ee
and an element $R\in\g\wedge\g$ is said to be a \textit{factorizable r-matrix} if it satisfies the modified classical Yang-Baxter equation,
\be\label{modYB}
\g\wedge\g\wedge\g\owns [R,R] + \phi=0.
\ee
\end{definition}
It is convenient to write the Yang-Baxter equation via the identification of $\g$ with $\g^*$, so that $R\in End(\g)$, and the skew-symmetry of $R$ becomes $R^t=-R$. Then the Yang-Baxter condition is written
\be
\langle X, [RY,RZ]\rangle + c.p. + \langle X,[Y,Z]\rangle =0\qquad\forall X,Y,Z\in\g,
\ee
or, equivalently,
\be
[RX,RY] - R\bigl([RX,Y]+[X,RY]\bigr) + [X,Y]=0 \qquad\forall X,Y\in\g.
\ee
In what follows the same symbol $R$ will be used indiscriminately for the element in $\g\wedge\g$ and for the corresponding element in $End(\g)$, in the expectation that the context will make all meanings clear.

\begin{remark}
In the present article, whenever we speak of an r-matrix on a Lie algebra $\g$, it will be assumed that $\g$ has a non-degenerate, invariant inner product, and the r-matrix is assumed to be factorizable.
\end{remark}

The following proposition follows directly from the Yang-Baxter condition.
\begin{proposition}
The two subspaces 
\[
\g_\pm=Im(R\pm I\!d)
\]
are both subalgebras. 
\end{proposition}

It is convenient to adopt the convention, for any $X\in\g$,
\be
X = X_+ - X_-,\qquad\hbox{with }\quad  X_{\pm}=\frac12(RX\pm X),  \quad X_{\pm}\in\g_\pm,
\ee
so that
\be
RX = X_+ + X_-.
\ee
We will denote by $G_\pm$ the subgroups of $G$ for which $Lie(G_\pm)=\g_\pm$.

A classical r-matrix can be used to define a Poisson bracket on $G$. First of all, we need to define the left- and right-derivatives of a function on $G$.

\begin{definition}
For $F\in Fun(G)$ the left-derivative $DF:G\to\g^*$ and the right-derivative $D'F:G\to\g^*$ are given by
\be
\dt F(e^{tX}ge^{tY}) = \langle D_gF,X\rangle + \langle D'_gF,Y\rangle\qquad\forall X,Y\in\g.
\ee
\end{definition}

As can easily be checked,
\be
D_gF = Ad^*_g D_g'F.
\ee

\begin{definition}
The \textit{Sklyanin bracket} $\{\,\cdot\,,\,\cdot\,\} : \cinf(G)\wedge\cinf(G)\to\cinf(G)$ is defined by
\be
\{F,H\} = \langle DF,R(DH)\rangle - \langle D'F,R(D'H)\rangle.
\ee
\end{definition}

\begin{proposition}\label{Gpbdef}
The Sklyanin bracket is a Poisson bracket; that is, it satisfies the Jacobi identity. Moreover, with the Sklyanin bracket, $G$ is a Poisson Lie group.
\end{proposition}

\begin{definition}
Suppose that $G$ is a Poisson Lie group, and let $A\subset G$ be a subgroup in $G$. $A$ is said to be a \textit{Poisson subgroup} if $A$ is also a Poisson subspace in $G$.
\end{definition}

\begin{remark}
A Poisson subgroup of a Poisson Lie group $G$ is itself a Poisson Lie group, with Poisson structure inherited from $G$.
\end{remark}

\begin{remark}
Suppose that $\g$ has an r-matrix. It is not necessarily the case that $\g_+\cap\g_-=\{0\}$: if this condition does hold then $\g$ can be decomposed as the vector space sum $\g=\g_+ + \g_-$, and, with respect to this decomposition, $R$ can be represented as the difference, $R = P_+ - P_-$, of projectors $P_{\pm}= \half(R\pm1)$, and the groups $\exp(\g_\pm)$ are Poisson subgroups.
\end{remark}

\subsection{The Drinfeld double of a Poisson Lie group}
In general, the Lie algebra of a Poisson Lie group may have no non-degenerate, invariant inner product, and is not required to have an r-matrix.  However, as we observe in this subsection, a Poisson Lie group $G$ can always be enlarged to a larger Poisson Lie group $D$ on which the Poisson structure is the Sklyanin bracket defined by an r-matrix, and inside which $G$ is a Poisson Lie subgroup.

\begin{definition}
Let $(G,\{\darg\})$ be a Poisson Lie group. The corresponding \textit{Drinfeld double} Lie algebra is the vector space direct sum $\g\oplus\g^*$, with Lie bracket (uniquely) defined by the properties\\
(i) $(\g,[\darg])$ and $(\g^*,[\ \cdot\ ,\ \cdot\ ]_*)$ are subalgebras,\\
(ii) the natural inner product on $\g\oplus\g^*$ is invariant; that is
\be
\langle [X,Y],Z\rangle = \langle X,[Y,Z]\rangle,\quad\forall X,Y,Z\in\g\oplus\g^*,
\ee
where, for $X=(x,\xi)$, $Y=(y,\eta)$,
\be
\langle X,Y\rangle = \xi(y)+\eta(x).
\ee
\end{definition}
Let $D$ be any connected Lie group with Lie algebra $\g\oplus\g^*$ such that the inclusion $\g\to\g\oplus\g^*$ integrates to a Lie group homomorphism $G\to D$, and take $G^*$ to be the connected Lie subgroup of $D$ with Lie algebra $\g^*$:
$G^*$ is called the `dual group', and $D$ is called the Drinfeld double group. The structures on $G^*$ and $D$ originate from $G$ with its Poisson structure as well as its group structure.\footnote{In general, given a Poisson Lie group $G$, there is no canonical dual group $G^*$, nor Drinfeld double $D$ of $G$. There may be several choices of $D$ up to covering, and there is no canonical choice of $G^*$ unless it is taken to be the connected and simply connected one. A detailed discussion can be found in \cite{LuYak}.}
Alternatively, we might view $D$ as the seed object, from which all other groups and spaces are derived: the natural pairing on the Drinfeld double Lie algebra is always non-degenerate and invariant, and as it is the direct sum of the subalgebras $\g\oplus\{0\}$ and $\{0\}\oplus\g^*$ it has a skew-symmetric r-matrix, and the Sklyanin Poisson structure defined by this r-matrix makes $D$ into a Poisson Lie group, in which $G$ and $G^*$ are Poisson subgroups.\footnote{For this reason, there is no loss of generality in approaching the study of Poisson Lie groups via the use of skew-symmetric r-matrices.} 

\subsubsection{The factorizable case}
\begin{definition}
Suppose that we have a Poisson Lie group $G$, with a non-degenerate invariant inner product on $\g$ and a skew-symmetric r-matrix $R\in\g\wedge\g$, satisfying \eqref{modYB}. The \textit{double} of $G$ is $G\times G$, and the \textit{double Lie algebra} is $\fd :=Lie(D)=\g\oplus\g$. (This an example of a Drinfeld double.) We introduce 
the inner-product $\llangle\darg\rrangle$ on $\fd$ defined by
\be\label{doubleinnprod}
\llangle (X_1,Y_1),(X_2,Y_2)\rrangle = \langle X_1,X_2\rangle - \langle Y_1,Y_2\rangle.
\ee
\end{definition}

The double Lie algebra can be decomposed as the sum of two subalgebras
\be
\g\oplus\g = \g^\delta \oplus \g_R.
\ee
Here $\g^\delta\cong\g$ is the diagonal subalgebra
\be
\g^\delta =\{(X,X)\ |\ X\in\g\},
\ee
and the other subalgebra,
\be
\ba
\g_R &= \{(X,Y)\in\g \oplus \g \ |\ P_-(X) = P_+(Y) \}\\
&=
\{(P_+(X),P_-(X))\ |\ X\in\g\},
\ea
\ee
is the image of the map $P_+\oplus P_- :$ $\g\to\g\oplus\g$. This decomposition gives rise to the double r-matrix on $\g\oplus\g$
\be
\hat R = P_{\g^\delta} - P_{\g_R}.
\ee
Several important properties of the double with the above inner product structure and r-matrix are summarised in the following
\begin{proposition}
The inner-product $\llangle\darg\rrangle$ on $\fd$ is invariant and non-degenerate. The double r-matrix $\hat R$ on the double of $\g$ is skew-symmetric with respect to $\llangle\darg\rrangle$, and is represented in terms of the r-matrix $R$ on $\g$ by the formula
\be
\hat R(X,Y) = \bigl( R(Y-X)\, , \, R(Y-X) \bigr) + (Y,X).
\ee
\end{proposition}
Just as $\langle\darg\rangle$ is used without comment to identify $\g$ with $\g^*$, by means of which the left- and right-derivatives of smooth functions on $G$ are viewed as being situated in $\g$, and by means of which $R$ is viewed indiscriminately as an element of $End(\g)$ or of $\g\wedge\g$; so is $\llangle\darg\rrangle$ used to identify $\fd\sim\fd^*$, thence to represent in $\fd$ left- and right-derivatives of smooth functions and to view $\hat R$ indiscriminately as an element of $End(\fd)$ or of $\fd\wedge\fd$.

Comparing the last collection of information with the general theory, $\g_R$ is isomorphic to the dual Lie algebra $(\g^*, [\darg]_*)$ and $\fd$ is isomorphic to the Drinfeld double Lie algebra.

\subsection{The Lu-Weinstein example}\label{LuWeinstein}
 
This subsection deals with a special case of the Lu-Weinstein example, from  \cite{LW}. Here, the base field is $\RR$ and all functions are real-valued functions.
Consider $\g=gl(n,\CC)$, viewed as a real Lie algebra. The inner product on $\g$ given by
\be
\langle X,Y\rangle = Im\,\tr(XY)
\ee
is invariant and non-degenerate. Let $\k$ and $\b$ be the two subalgebras in $\g$:
\be
\ba
&\k = u(n) = \{X\in\g\ | \ X+X^\dag=0\},\\
&\b = \{\hbox{upper-triangular matrices with real entries on the diagonal}\}.
\ea
\ee
It is easy to check that, as vector spaces, $\g=\k + \b$; any element in $\g$ can be written uniquely as the sum of elements in $\k$ and in $\b$. Also, with respect to $\langle\darg\rangle$, $\k^\perp=\k$ and $\b^\perp=\b$. It follows that 
\be
R = P_\k - P_\b
\ee
defines an r-matrix on $\g$, skew-symmetric with respect to $\langle\darg\rangle$. Groups $G$, $K$, and $B$ corresponding to $\g$, $\k$, and $\b$ are
\[
\ba
G &= \GL(n,\CC),\\
K &= \U(n),\\
B &= \{\hbox{upper triangular matrices with real, positive entries on the diagonal}\};
\ea
\]
$G$ is the Drinfeld double of $K$, with $B$ identified with $K^*$.

Going through the definitions, with the identification $\g^*\cong\g$ inducing $\k^*\cong\b$ and $\b^*\cong\k$, by means of the decomposition $G\owns g =kb^{-1}$, we get the Poisson brackets on $K$ and $B$ inherited from the Sklyanin PB on $G$:
\be\label{KPb}
\{F_1,F_2\}(k) = Im\,\tr \bigl(D_kF_1\, kD'_kF_2k^{-1}\bigr) ,\qquad F_1,F_2\in\cinf(K,\RR),
\ee

\be
\{H_1,H_2\}(b) = - Im\,\tr\bigl( D_bH_1\, bD'_bH_2b^{-1}\bigr) ,\qquad H_1,H_2\in\cinf(B,\RR).
\ee
Here, the left- and right-derivatives are given by
\be
\ba
F\in\cinf(K)\Rightarrow &D_kF,D'_kF\in\k^*\cong\b^\perp=\b:\\ 
&\dt F(e^{t\xi}ke^{t\eta}) = Im\,\tr\bigl(\xi D_kF + \eta D'_kF\bigr)\qquad\forall\xi,\eta\in\k,\\
H\in\cinf(B)\Rightarrow &D_bH,D'_bH\in\b^*\cong\k^\perp=\k:\\ 
&\dt H(e^{t\alpha}be^{t\beta}) = Im\,\tr\bigl(\alpha D_bH + \beta D'_bH\bigr)\qquad\forall\alpha,\beta\in\b.
\ea
\ee
An alternative representation of $B$ comes from the invertible map $B\to \hermp = \{$positive definite, hermitian matrices$\}$:
\be\label{identifyBP}
B\owns b\mapsto bb^\dag=\Lm\in \hermp,
\ee
in terms of which the Poisson bracket is
\be
\{F,H\}(\Lambda) = 8 Im\,\tr\bigl(\Lm X\,(\Lm Y)_\k\Bigr),
\ee
where the skew-hermitian matrices $X=d_\Lm F$ and $Y=d_\Lm H$ are the derivatives of $F$ and $H$, given by
\be
\dt F(\Lm+tA) = Im\,\tr\bigl(A\,d_\Lm F\bigr)\qquad\forall\ \hbox{hermitian}\ A,
\ee
and the subscript `$\k$' indicates projection onto $\k$ parallel to $\b$.

Equivalent to the above formulae for Poisson brackets, are those for Hamiltonian vector fields:
\be
\ba
\XX_F(k) &\sim \exp\bigl(t[k D_k'Fk^{-1}]_\k\bigr)k\qquad F\in\cinf(K,\RR),\\
\XX_F(b) & \sim \exp\bigl(-t[b D_b'Fb^{-1}]_\b\bigr)b \qquad\ F\in\cinf(B,\RR), \\ 
\XX_F(\Lm) &= 4 \left[\bigl(\Lm d_\Lm F)_\k\,,\,\Lm\right]\qquad\qquad F\in\cinf(\hermp,\RR).
\ea
\ee

\subsection{A complex (holomorphic) example}\label{holex}

The example in this subsection may be generalised to any reductive complex Lie algebra; see for example \cite{ES}, \cite{Yred}, \cite{H}.
Here, the base field is $\CC$ and all functions are holomorphic functions.
Viewing $\g=gl(n,\CC)$ as a complex Lie algebra, we make use of the invariant non-degenerate inner product on $\g$,
\be
\langle X,Y\rangle = \tr(XY).
\ee
Let $\n_\pm$, $\h$ be the subalgebras in $\g$:
\[
\ba
\n_+ &= \{\hbox{ strictly upper triangular matrices }\},\\
\n_- &= \{\hbox{ strictly lower triangular matrices }\},\\
\h &= \{\hbox{ diagonal matrices }\}.
\ea
\]
Denote by $N_\pm, H$, the corresponding subgroups in $\GL(n,\CC)$.
Clearly, $\g = \n_+ + \h + \n_-$; that is, any element of $\g$ may be written uniquely as the sum of elements in $\n_+$, in $\h$, and in $\n_-$. It may be checked that, using the projectors $P_U$, $P_L$, $P_\Delta$, with obvious meanings ($U$ stands for `upper', $L$ for `lower' and $\Delta$ for `diagonal'), the linear map
\be
R = P_U - P_L
\ee
is an r-matrix on $\g$, skew-symmetric with respect to $\langle\darg\rangle$. We have
\be
\g\cong \g^\delta = \{(X,X)\,|\, X\in\g\}\subset\g\oplus\g,
\ee
and
\be
\g_R = \{ (x+p,y-p)\, | \, x\in\n_+,\ y\in\n_-,\ p\in\h\}.
\ee

In a dense open subset, the Drinfeld double
\be
D= G\times G \ \  = \{(a,b)\,|\, a,b\in G\}
\ee
can be represented as the product of the two subgroups
\be
D\supset G^\delta = \{(g,g)\,|\, g\in G\}
\ee
and
\be
D\supset G_R= \{ (x\Gamma,y\Gamma^{-1})\,|\, x\in N_+, \ y\in N_-, \  \Gamma\in H\};
\ee
thus, for a dense open subset $D^o\subset D$, we may write
\be\label{DrepGG}
D^o \owns(a,b) = (g_+,g_-)(x,x),
\ee
with 
\be
g:= g_+g_-^{-1} = ab^{-1}
\ee
and
\be
x = (ab^{-1})_+^{-1} a = (ab^{-1})_-^{-1}b.
\ee
It should be noted that the representation \eqref{DrepGG} is not unique, since $G_R\cap G^\delta$ is non-trivial; indeed, it consists of the pairs $(\Delta,\Delta)$ with $\Delta$ a diagonal matrix all of whose entries are in $\{1,-1\}$. Another way to say this is that $G^\delta\times G_R\to D$ is a $(\ZZ_2{})^n$ covering. This causes no problem as the product map $G^\delta\times G_R\to D^o$ is a local diffeomorphism.

The following is self-evident, but it is needed later on, so it is worth having it stated emphatically here:

\begin{proposition}\label{GRactioninG}
The map $(g_+,g_-)\mapsto g_+g_-^{-1}$ from $G_R$ to $G$ represents the left-mutiplication action of $G_R$ on itself in the form
\be
g\cdot h = g_+hg_-^{-1}
\ee

\end{proposition}

\subsection{Poisson actions}
\begin{definition}
Let $(M, \{\ddarg\}_M)$ be a Poisson space, let $(G, \{\ddarg\}_G)$ be a Poisson Lie group, and suppose that $G$ acts on $M$. The action is said to be a \textit{Poisson action} if the natural map $G\times M\to M$ defined by this action is Poisson; that is, if for all $F,H\in Fun(M)$
\be
\{F,H\}_M(g\cdot m) = \{F(\,\cdot\,m),H(\,\cdot\,m)\}_G(g) + \{F(g\,\cdot\,),H(g\,\cdot\,)\}_M(m).
\ee
\end{definition}
The notion of momentum map adapted to the Poisson Lie group context is due to Lu \cite{Lu}. Before stating it, we need to understand its component parts. The Lie algebra $\g$ is dual to the Lie algebra of the dual group; that is, $\g=\bigl(Lie(G^*)\bigr)^*$. An element of $\g$ defines a right-invariant one-form on $G^*$ as follows. For $X\in\g$, we must define the function $\langle\tilde X,v\rangle\in Fun(G^*)$, obtained by pairing the right-invariant one-form $\tilde X$ corresponding to $X$ with a vector-field $v\in Vect(G^*)$. Consider an arbitrary point $p\in G^*$, and find $\xi\in\g^*$ such that $T_pG^*\owns v(p)\sim \exp(t\xi)p$; then the value of the function $\langle\tilde X,v\rangle$ at the point $p$ is
\be
\langle\tilde X,v\rangle(p) = \langle X,\xi\rangle.
\ee

\begin{definition}
Suppose that $G$ is a Poisson Lie group, and $M$ is a Poisson space, with a Poisson action of $G$ on $M$. For any $X\in\g=Lie(G)$, denote by $\tilde X\in\Omega^1(G^*)$ the right-invariant one-form on $G^*$ defined by $X$, and denote by $X_M\in Vect(M)$ the infinitesimal action of $X$ on $M$. Denote by $P$ the Poisson tensor on $M$.
A map $J:M\to G^*$ is said to be a momentum map for the Poisson action of $G$ on $M$ if, for any $X\in\g$,
\be
P\bigl(\ \cdot\ , J^*(\tilde X) \bigr) = X_M .
\ee
\end{definition}
Pairing each side with the one-form $dF$, for an arbitrary function $F\in Fun(M)$, we obtain
\be
-\langle\tilde X,dJ\cdot \XX_F \rangle = -\bigl( J^*(\tilde X)\bigr)(\XX_F) = X_M(F).
\ee



\section{The semidirect product of Poisson $G$-spaces}\label{semidirectsec}

Suppose that $M$ and $N$ are Poisson spaces, and that $G$ is a Poisson Lie group, with dual $G^*$ and Drinfeld double $D$. Suppose that $D$ acts on $M$ and on $N$, and that both actions are Poisson actions: then the restrictions $G \circlearrowright M$, $G\circlearrowright N$, $G^*\circlearrowright M$, $G^*\circlearrowright N$ are also Poisson actions, as $G$ and $G^*$ are Poisson subgroups in $D$. 
 
\begin{definition}
Let $A$ be a $D$-space, and consider $F\in Fun(A)$. For any point $x\in A$ we define the \textit{action-derivatives} $D_xF\in\g^*=Lie(G^*)$ and $\nabla_xF\in\g=Lie(G)$ by
\be\label{actionderivatives}
\ba
&\dt F\bigl(e^{tX}\cdot x\bigr) = \langle D_xF,X\rangle\qquad\forall X\in\g,\\
&\dt F\bigl(e^{t\xi}\cdot x\bigr) = \langle\xi , \nabla_xF \rangle\qquad\forall \xi\in\g^*,
\ea
\ee
where angled brackets denote here the natural inner product on $\g\oplus\g^*=Lie(D)$.
\end{definition}

Assuming that the action of $G$ on $M$ has a momentum map $J:M\to G^*$, let us define the map $\S: M\times N\to M\times N$,
\be
\S : (m,n)\mapsto (m,J(m)n).
\ee
Clearly $\S$ is invertible, and it induces what we may call the semidirect Poisson structure on $M\times N$ described in the following
\begin{proposition}\label{semidirect}

For $f,h\in Fun(M)$ and $\varphi,\psi\in Fun(N)$, define the functions $F,H\in Fun(M\times N)$,
\be
\ba
F(m,n) &= f(m) + \varphi(n),\\
H(m,n) &= h(m) + \psi(n).
\ea
\ee

The semidirect Poisson structure on $M\times N$, obtained from the direct sum Poisson structure by applying the map $\S$, is given by
\be\label{semidirectformula}
\{F,H\}(m,n) = \{f,h\}_M(m) + \{\varphi,\psi\}_ N(n) + \langle D_mf,\nabla_n\psi\rangle - \langle D_mh,\nabla_n\varphi\rangle .
\ee
\end{proposition}
\begin{proof}
(i) The defining property of the momentum map,
\[
\langle dJ\cdot\XX_f,\tilde X\rangle = - X_M\righthalfcup df\qquad\forall X\in\g,
\]
is equivalent to
\be
T_{J(m)}G^*\owns dJ\cdot\XX_f (m) \sim \exp(-tD_mf)J(m).
\ee
(ii) The Poisson action property of $G^*\circlearrowright N$ is
\be
\ba
\{\varphi(J(m)\,\cdot\ ), \psi(J(m)\,\cdot\ )\}_N(n) + \{\varphi(\ \cdot\, n), \psi(\ \cdot\, n)\}_{G^*}(J(m)) = \{\varphi,\psi\}_N(J(m)n)
\ea
\ee
and, as $J:M\to G^*$ is a Poisson map, we have
\be
\{\varphi(\ \cdot\, n), \psi(\ \cdot\, n)\}_{G^*}(J(m)) = \{\varphi( J(\,\cdot\,)n), \psi( J(\,\cdot\,) n)\}_M(m).
\ee
Putting these together, writing $\tilde n=J(m)n$, so that $\S^*F(m,n)=F(m,\tilde n)$, $\S^*H(m,n)=H(m,\tilde n)$, and temporarily using $\{\ ,\ \}_\S$ to denote the modified Poisson bracket, we have
\[
\ba
\{F,H\}_\S(m,\tilde n) &= \{\S^*F,\S^*H\}(m,n)\\
&=
\{f,h\}_M(m) + \{\varphi,\psi\}_N(\tilde n)\\ 
&\qquad+ \dt\varphi\bigl(\exp(-tD_mh)J(m)n\bigr) - \dt\psi\bigl(\exp(-tD_mf)J(m)n\bigr)\\
&=
\{f,h\}_M(m) + \{\varphi,\psi\}_N(\tilde n) +  \langle D_mf,\nabla_{\tilde n}\psi\rangle - \langle D_mh,\nabla_{\tilde n}\varphi\rangle.
\ea
\]

\end{proof}

Another ingredient is the possibility to modify the Poisson structure in such a way as to break the Poisson symmetry of the Drinfeld double whilst preserving that of the Poisson subgroup $G$; that is, breaking the Poisson symmetry of $G^*$.

\begin{proposition}
Let $G$ be a Poisson Lie group, let $M$ be a Poisson space with Poisson structure $P$, and suppose that $G$ has a Poisson action on $(M,P)$. Let $\Pi\in Vect(M) \wedge Vect(M)$ be a $G$-invariant Poisson tensor on $M$, and suppose that the Schouten bracket $[P,\Pi]$ is zero. Then $\hat P=P+\Pi$ is a Poisson structure, and the action of $G$ on $M$ will still be a Poisson action with respect to $\hat P$.
\end{proposition}
The proof of this statement is self-evident. In general, if the Poisson action of $G$ results from restriction of the Poisson action of its double $D$, then it may be expected that the full $D$-symmetry is broken by the addition of $\Pi$, as there is no expectation that $\Pi$ be $G^*$-invariant.



\section{Examples of Poisson actions---following Zakrzewski}\label{examples}

Although the examples presented here can be found amongst the results of Zakrzewksi in \cite{Z}, they are presented from a slightly different, and hopefully illuminating, point of view. 

\subsection{The Poisson action of $\GL(n,\CC)$ on several copies of $\CC^n\times\CC^n$}\label{holss}
Let $V$ and $W$ be vector spaces over $\CC$. Suppose that $dim_\CC(V)=n$ and $dim_\CC(W)=d$. We look for a Poisson Lie covariant Poisson bracket on the space
\be
\ba
M &= Hom(W,V)\times Hom(V,W)\\
&=
(V\otimes W^*)\times(V\otimes W^*)^*.
\ea
\ee

Alternatively, after choosing bases in $V$ and $W$,
\be
M =
\{ m=(v,w^T)\,|\, v,w\in Mat_{n\times d}(\CC)\}.
\ee
The space $M$ will be treated as complex, so that all functions are complex-valued; we restrict however to holomorphic functions.

The Poisson Lie groups $G=\GL(V)$ and $H=\GL(W)$ act naturally on $M$:
\be
g\cdot (v,w^T) = (gv,w^Tg^{-1}),\quad\hbox{and}\quad h\cdot(v,w^T) = (vh^{-1}, hw^T)\qquad g\in G, h\in H.
\ee
This problem is most easily treated by extending to the action of the doubles of $G$ and $H$, and then restricting to $G$ and $H$ as Poisson subgroups. To this end we'll consider the actions of $G\times G$ and of $H\times H$ on $M$:
\be
(a,b)\cdot (v,w^T) = (av,w^Tb^{-1}),\ \hbox{and}\ (p,q)\cdot(v,w^T) = (vp^{-1}, qw^T)\quad a,b\in G,\ \  p,q\in H.
\ee

Thus we get maps from the doubles of the respective Lie algebras, $\g=Lie(G)= gl(V)$ and $\h=Lie(H)=gl(W)$, to $Vect(M)$: 
\[
\ba
&\g\oplus\g \owns (A,B)\mapsto ( A, B)_M\in Vect(M); \ \ 
T_mM\owns( A, B)_M(m)\sim (\exp(tA),\exp(tB))\cdot m,\\
&\h\oplus\h \owns (P,Q)\mapsto ( P, Q)_M\in Vect(M); \ \ 
T_mM\owns( P, Q)_M(m)\sim (\exp(tP),\exp(tQ))\cdot m,
\ea
\]
given explicitly by
\be
( A, B)_M(v,w^T) = (Av, -w^TB)
\ee
and
\be
( P, Q)_M(v,w^T) = (-vP,Qw^T).
\ee
We will use the invariant non-degenerate inner products on $\g$ and $\h$ given by the trace forms:
\be
\ba
\langle X_1,X_2\rangle = \tr(X_1X_2)\qquad X_1,X_2\in\g,\\
\langle A_1,A_2\rangle = \tr(A_1A_2)\qquad A_1,A_2\in\h.\\
\ea
\ee
Suppose that $R\in\g\wedge\g$ and $\rho\in\h\wedge\h$ are r-matrices on the respective Lie algebras, and identify them with maps in $End(\g)$ and $End(\h)$ by means of respective traces. With this convention, we may write the Yang-Baxter conditions on $R$ and $\rho$ either as
\be
\ba
\langle [R,R], (X,Y,Z)\rangle &= - \langle X, [Y,Z]\rangle \qquad\forall X,Y,Z\in\g,\\
\langle [\rho,\rho] , (A,B,C)\rangle &= - \langle A, [B,C]\rangle \qquad\forall A,B,C\in\h,
\ea
\ee
or as
\be
\ba
{}[RX,RY] - R\bigl([RX,Y]+[X,RY]\bigr) = -[X,Y],\\
[\rho A,\rho B] - \rho\bigl([\rho A,B] + [A,\rho B]\bigr) = -[A,B].
\ea
\ee
The r-matrices $R$ and $\rho$ on $\g$ and on $\h$ give rise to r-matrices $\hat R$ and $\hat\rho$ on the respective doubles:
\be
\ba
\hat R(A,B) &= (R(B-A)+B,R(B-A)+A),\\
\hat\rho(P,Q) &=(\rho(Q-P)+Q,\rho(Q-P)+P).
\ea
\ee
We define the contravariant two-tensors (bi-vectors) $\hat R_M,\hat\rho_M\in Vect(M)\wedge Vect(M)$ by means of the maps from $\g,\h$ to $Vect(M)$, as follows. Suppose that
\be
\hat R = \sum_i(A_i,B_i)\otimes (X_i,Y_i), \qquad \hat\rho = \sum_i (P_i,Q_i)\otimes (U_i,V_i);
\ee
then
\be
\hat R_M = \sum_i(A_i,B_i)_M\otimes (X_i,Y_i)_M, \qquad \hat\rho = \sum_i (P_i,Q_i)_M\otimes (U_i,V_i)_M.
\ee
We have 
\be
T_m^*M \cong \{(\xi,x^T)\,|\, \xi,x\in Hom(W,V)\},
\ee
and as $M$ is a vector space, statements about Poisson brackets and Hamiltonian vector fields and so on can all be made in terms only of linear functions.
It is straightforward to work through the notations and deduce the following
\begin{proposition}
Let $(A,B)\in\g\oplus\g$ and $(P,Q)\in\h\oplus\h$. Then, at the point $m=(v,w^T)\in M$, the vector fields $(A,B)_M$ and $(P,Q)_M$ define the following forms on $T_m^*M$
\[
\ba
(A,B)_M(v,w^T)(\xi,x^T) &= \langle A, vx^T\rangle  - \langle B,\xi w^T\rangle = \llangle (A,B),(vx^T,\xi w^T)\rrangle,\\
(P,Q)_M(v,w^T)(\xi,x^T) &= - \langle P, x^Tv\rangle + \langle Q, w^T\xi\rangle = -\llangle(P,Q),(x^Tv,w^T\xi)\rrangle,
\ea
\]
\end{proposition}
\noindent
and hence
\begin{proposition}
At the point $m=(v,w^T)\in M$, the tensors $\hat R_M$ and $\hat\rho_M$ define the forms on $T_m^*M$
\[
\ba
\hat R_M(v,w^T)\bigl( (\xi,x^T),(\eta,y^T)\bigr) 
&=
- \tr\Bigl( (vx^T-\xi w^T)R(vy^T - \eta w^T ) + w^Tv(y^T\xi-x^T\eta)\Bigr),\\
\hat\rho_M(v,w^T)\bigl( (\xi,x^T),(\eta,y^T)\bigr) &=
-\tr\Bigl( (x^Tv-w^T\xi) \rho(y^Tv-w^T\eta) + vw^T(\xi y^T-\eta x^T)\Bigr).
\ea
\]
\end{proposition}

\begin{proof}
Using the map $(A,B)\mapsto(A,B)_M$ of $\g\oplus\g$ to $Vect(M)$ from the previous proposition, we compute
\[
\ba
\hat R_M(v,w^T)\bigl( (\xi,x^T),(\eta,y^T)\bigr) &=
\langle \hat R , (vx^T,\xi w^T)\otimes(vy^T,\eta w^T)\rangle\\
&=
\llangle (vx^T,\xi w^T),\hat R(vy^T,\eta w^T) \rrangle\\
&=
\left\langle\!\left\langle \bigl(vx^T,\xi w^T\bigr),\bigl(R(\eta w^T-vy^T) + \eta w^T , R(\eta w^T-vy^T) + vy^T\bigr)  \right\rangle\!\right\rangle \\
&=
\langle vx^T-\xi w^T, R(\eta w^T-vy^T) \rangle + \langle vx^T,\eta w^T\rangle - \langle \xi w^T, vy^T\rangle\\
&=
\tr\Bigl( - (vx^T-\xi w^T)R( vy^T - \eta w^T) + w^Tv(x^T\eta - y^T\xi)\Bigr).
\ea
\]
The formula for $\hat\rho_M$ can be obtained in a similar fashion.

\end{proof}

Because the actions of $G\times G$ and $H\times H$ on $M$ commute with one another, the Schouten bracket between $\hat R_M$ and $\hat \rho_M$ is zero, and because the maps $X\mapsto X_M$ and $A\mapsto A_M$ from $\g$ and $\h$ to $Vect(M)$ are Lie algebra homomorphisms, we have
\be
\ba
{}[\hat R_M,\hat R_M] &=[\hat R,\hat R]_M,\\
[\hat \rho_M,\hat \rho_M] &=[\hat \rho,\hat\rho]_M,
\ea
\ee
or, more explicitly, we may formulate these statements as
\begin{proposition}\label{Rrhobrackets}
The Schouten brackets between the contravariant two tensors $\hat R_M$ and $\hat \rho_M$ can be written in the form
\[
\ba
{} [\hat R_M,\hat\rho_M] & =0,\\
[\hat R_M,\hat R_M](v,w^T)\bigl( (\xi,x^T),(\eta,y^T),(\zeta,z^T)\bigr)
&=
- \left\langle\!\left\langle (vx^T,\xi w^T)\,,\,\bigl[ (vy^T,\eta w^T),(vz^T,\zeta w^T)\bigr] \right\rangle\!\right\rangle, \\
[\hat\rho_M,\hat\rho_M](v,w^T)\bigl( (\xi,x^T),(\eta,y^T),(\zeta,z^T)\bigr)
&=
+ \left\langle\!\left\langle (x^Tv , w^T\xi)\,,\,\bigl[ (y^Tv , w^T\eta),(z^Tv , w^T\zeta)\bigr] \right\rangle\!\right\rangle.
\ea
\]
\end{proposition}
The signs involved in the last two formulae arise from the signs involved in the different actions (both actions are ``left-actions'').

An easy corollary is that the sum $P = \hat R_M + \hat\rho_M$ is a Poisson tensor on $M$: expanding the righthand sides in Proposition \ref{Rrhobrackets}, this follows by applying the property $\tr(AB^T) = \tr(B^TA)$ for any rectangular matrices $A, B$ of the same shape.
We can do more though. Any bivector on $M$, invariant with respect to the actions of $G$ and of $H$ on $M$, will have zero Schouten bracket with $P$, so we may add (any multiple of) the invariant Poisson bivector $\Pi$, given by
\be\label{invformhol}
\Pi(v,w^T)\bigl( (\xi,x^T),(\eta,y^T)\bigr) = \tr( x^T\eta - y^T\xi),
\ee
and still retain the Poisson property.
\begin{proposition}
For any $\lambda\in\CC$, the bi-vector
\[
P_\lambda = \hat R_M + \hat\rho_M + 2\lambda\Pi
\]
is a Poisson tensor on $M$, and the actions of $G$ and of $H$ on $M$ are Poisson actions when $G$ and $H$ have the respective Sklyanin brackets defined by $R$ and $\rho$. For $\lambda=0$, the actions of the Drinfeld doubles $G\times G$ and $H\times H$ on $M$ are Poisson actions when $G\times G$ and $H\times H$ have the respective Sklyanin brackets defined by $\hat R$ and $\hat\rho$. 

\end{proposition}
The first part of this result is a consequence of the fact that $G$ and $H$ are Poisson subgroups (via diagonal embedding) of their respective doubles $G\times G$ and $H\times H$. It is convenient to use $(\half)$ times $P$ to define the Poisson bracket.

\begin{proposition}\label{ZakPBhol}
The bracket $\{\ddarg\}$ on $M$, defined for linear functions $F,H\in\hol(M,\CC)$
\[
F(v,w^T) = \tr(x^Tv+w^T\xi),\quad H(v,w^T) = \tr(y^Tv+w^T\eta),
\]
by
\[
\ba
\{F,H\}(v,w^T) &= \half\tr\Bigl( - (vx^T-\xi w^T)R(vy^T - \eta w^T ) - (x^Tv-w^T\xi) \rho(y^Tv-w^T\eta)\\ 
&\qquad\qquad - w^Tv(y^T\xi-x^T\eta) - vw^T(\xi y^T-\eta x^T) + 2\lambda( x^T\eta - y^T\xi) \Bigr),
\ea
\]
defines a Poisson bracket on $M$. Equivalently, the Hamiltonian vector field corresponding to $F$ is
\[
\ba
&\XX_F(v,w^T) = \\
&\qquad\half\Bigl( - R(vx^T - \xi w^T)v - v\rho(x^Tv - w^T\xi) + \xi w^Tv + vw^T\xi + 2\lambda\xi,\\
&\qquad\qquad\qquad
w^T R(vx^T - \xi w^T) + \rho(x^Tv - w^T\xi)w^T - w^Tvx^T - x^Tvw^T - 2\lambda x^T\Bigr).
\ea
\]
The actions of $G$ and $H$ on $M$, given by
\[
g\cdot(v,w^T) = (gv,w^Tg^{-1}),\quad h\cdot (v,w^T) = (vh^{-1}, hw^T),
\]
are Poisson actions when $G$ is endowed with the Sklyanin bracket defined by $R\in\g\wedge\g$ and $H$ is endowed with the Sklyanin bracket defined by $\rho\in\h\wedge\h$.
\end{proposition}

Momentum maps for the Poisson actions of $G$ and $H$ should be maps from $M$ to the respective dual groups $G_R$ and $H_\rho$. If $\lambda\neq0$, we define the maps $\Phi^G_\lambda:M\rightarrow G$, and $\Phi^H_\lambda:M\rightarrow H$, by
\be\label{holmomGH}
\Phi_\lambda^G(v,w^T) = \lambda Id_n + vw^T,\qquad \Phi_\lambda^H(v,w^T) = \bigl(\lambda Id_d + w^Tv\bigr)^{-1},
\ee
but these must be treated as maps to $G_R$ and to $H_\rho$ via the identifications 
\be\label{factG}
G_R\sim G;\quad (g_+,g_-)\mapsto g_+g_-^{-1},
\ee
and
\be\label{factH}
H_\rho\sim H;\quad (h_+,h_-)\mapsto h_+h_-^{-1}.
\ee
Care must be taken to respect the fact that $\Phi^G$ and $\Phi^H$ are both defined only on open subsets in $M$; this is because \eqref{holmomGH} must represent invertible matrices, and the factorisations \eqref{factG} and \eqref{factH} must be possible, all of which are conditions restricting $(v,w^T)$ to lie in some open subset in $M$.

We will prove
\begin{proposition}
For $\lambda\neq0$ the maps $\Phi_\lambda^G$ and $\Phi_\lambda^H$ are momentum maps for the Poisson actions on $M$ of $G$ and $H$ respectively.
\end{proposition}

\begin{proof}

Consider, for $F\in Hol(M)$ with $x^T=\delta_vF$ and $\xi=\delta_{w^T}F$,
\[
\ba
2\left(d\Phi_\lambda^G\cdot\XX_F\right)(v,w^T)
&= 
\Bigl( - R(vx^T-\xi w^T)v - v\rho(x^Tv-w^T\xi) + 2\lambda\xi  + vw^T\xi  + \xi w^Tv\Bigr)w^T\\
&\qquad+ v\Bigl( w^TR(vx^T-\xi w^T) + \rho(x^Tv-w^T\xi)w^T - 2\lambda x^T - x^Tvw^T - w^Tvx^T\Bigr)\\
&=
- 2\lambda(vx^T-\xi w^T) - \bigl((R+Id)(vx^T-\xi w^T)\bigr)vw^T - vw^T\bigl((R-1)(vx^T-\xi w^T)\bigr)\\
&=
-\bigl((R+Id)(vx^T-\xi w^T)\bigr)[\lambda Id + vw^T] + [\lambda Id + vw^T]\bigl((R-1)(vx^T-\xi w^T)\bigr)\\
&=
- \bigl((R+Id)(vx^T-\xi w^T)\bigr)\Phi_\lambda^G(v,w^T) + \Phi_\lambda^G(v,w^T)\bigl((R-1)(vx^T-\xi w^T)\bigr).\\
\ea
\]
By means of the result stated in proposition \ref{GRactioninG}, we have
\[
\left(d\Phi_\lambda^G\cdot\XX_F\right)(v,w^T)
\sim
\exp(-t(vx^T-\xi w^T))\cdot\Phi_\lambda^G(v,w^T),
\]
so that, for $X\in\g$,
\[
\ba
\left\langle\tilde X\,,\, \left(d\Phi_\lambda^G\cdot\XX_F\right)(v,w^T)\right\rangle
&=
-\langle X\,,\, (vx^T-\xi w^T)\rangle\\
&=
-\tr\bigl( x^T(Xv) - (w^TX)\xi\bigr)\\
&=
\dt F\bigl(\exp({-tX})\cdot(v,w^T)\bigr)\\
&=
- \langle X_M, dF\rangle(v,w^T);
\ea
\]
that is, for any $F\in Hol(M)$,
\[
\left\langle\tilde X\,,\, \left(d\Phi_\lambda^G\cdot\XX_F\right)\right\rangle
=
- \langle X_M, dF\rangle,  \qquad X\in\g,
\]
as required.

\medskip
\noindent
Now consider
\[
\ba
&2\left(d\Phi_\lambda^H\cdot\XX_F\right)(v,w^T) \\
&=
- \bigl(\lambda Id + w^Tv\bigr)^{-1} \left[w^T\Bigl( - R(vx^T-\xi w^T)v - v\rho(x^Tv-w^T\xi) + 2\lambda \xi + vw^T\xi + \xi v\Bigr)\right.\\
&\qquad\left.+\Bigl( w^TR(vx^T-\xi w^T) + \rho(x^Tv-w^T\xi)w^T - 2\lambda x^T - x^Tvw^T - w^Tvx^T\Bigr)v\right] \bigl(\lambda Id + w^Tv\bigr)^{-1}\\
&=
- \bigl(\lambda Id + w^Tv\bigr)^{-1} \Bigl[ - 2\lambda(x^Tv-w^T\xi ) \\
&\qquad\qquad + \bigl((\rho-Id)(x^Tv-w^T\xi)\bigr)w^Tv - w^Tv\bigl((\rho+1)(x^Tv-w^T\xi)\bigr) \Bigr]\bigl(\lambda Id + w^Tv\bigr)^{-1} \\
&=
- \bigl(\lambda Id + w^Tv\bigr)^{-1} \Bigl[\bigl((\rho-Id)(x^Tv-w^T\xi)\bigr)(\lambda Id + w^Tv)\\ 
&\qquad\qquad\qquad\qquad\qquad\qquad
- (\lambda Id + w^Tv)\bigl((\rho+1)(x^Tv-w^T\xi ) \bigr)\Bigr] \bigl(\lambda Id + w^Tv\bigr)^{-1} \\
&=
\bigl((\rho+1)( x^Tv - w^T\xi )\bigr)\Phi_\lambda^H(v,w^T)  -  \Phi_\lambda^H(v,w^T) \bigl((\rho-1)( x^Tv - w^T\xi )\bigr).
\ea
\]
By means of the result stated in Proposition \ref{GRactioninG}, we have
\[
\left(d\Phi_\lambda^H\cdot\XX_F\right)(v,w^T)
\sim
\exp(t(x^Tv - w^T\xi ))\cdot\Phi_\lambda^H(v,w^T),
\]
so that, for $X\in\h$, 
\[
\ba
\left\langle\tilde X\,,\, \left(d\Phi_\lambda^H\cdot\XX_F\right)(v,w^T)\right\rangle
&=
\langle X\,,\, (x^Tv - w^T\xi )\rangle\\
&=
\tr\bigl( x^T(vX) - (Xw^T)\xi\bigr)\\
&=
\dt F\bigl(\exp({-tX})\cdot(v,w^T)\bigr)\\
&=
-\langle X_M, dF\rangle(v,w^T);
\ea
\]
that is, for any $F\in Hol(M)$,
\[
\left\langle\tilde X\,,\, \left(d\Phi_\lambda^H\cdot\XX_F\right)\right\rangle
=
-\langle X_M, dF\rangle,  \qquad X\in\h,
\]
as required.

\end{proof}

For later use, it is also useful to give attention to computation of the action-derivatives.

\begin{proposition}\label{actionderivshol}
Let $f\in\hol(M)$. The action-derivatives $D_{(v,w^T)}f\in\g_R$ and $\nabla_{(v,w^T)}f\in\g^\delta$, defined in \eqref{actionderivatives}, are 
\[
\ba
D_{(v,w^T)}f &=
\half\bigl( R(vx^T - \xi w^T), R(vx^T - \xi w^T)\bigr) + \half\bigl(vx^T - \xi w^T, \xi w^T - vx^T\bigr),\\
\nabla_{(v,w^T)}f
&=
\half\bigl(vx^T + \xi w^T, vx^T + \xi w^T\bigr) + \half\bigl(R(\xi w^T - vx^T)\,,\, R(\xi w^T - vx^T)\bigr),
\ea
\]
where $x^T=\delta_vf$ and $\xi=\delta_{w^T}f$.
\end{proposition}
\begin{proof}
It is sufficient to prove the claim when $f$ is the linear function,
\[
f(v,w^T) = \tr(x^Tv + w^T\xi),
\]
with $x,\xi\in\CC^n$ constant.
We compute
\[
\ba
\dt f\bigl(e^{t(A,A)}\cdot(v,w^T)\bigr) &= \tr\bigl(x^TAv - w^TA\xi\bigr)\quad\forall (A,A)\in\g^\delta\\ 
\Rightarrow\ 
\bigl\langle D_{(v,w^T)}f \,,\, (A,A)\bigr\rangle &= \tr A(vx^T - \xi w^T) \\
\Rightarrow\ 
\g_R\owns D_{(v,w^T)}f &= \half\bigl( (R+1)(vx^T - \xi w^T), (R-1)(vx^T - \xi w^T)\bigr).
\ea
\]
Similarly
\[
\ba
\dt f\bigl(e^{t(A_+,A_-)}
\cdot(v,w^T)\bigr) &= \tr\bigl(x^TA_+v - w^TA_-\xi\bigr)\quad\forall(A_+,A_-)\in\g_R\\
\Rightarrow\
\bigl\langle\nabla_{(v,w^T)}f \,,\,(A_+,A_-)\bigr\rangle &= \tr\bigl(A_+vx^T-A_ - \xi w^T\bigr)\\
\Rightarrow\ \g^\delta\owns  \nabla_{(v,w^T)}f &=\half\bigl( (1-R)(vx^T) + (1+R)(\xi w^T)\,,\, (1-R)(vx^T) + (1+R)(\xi w^T)\bigr).
\ea
\]

\end{proof}

\subsection{The Poisson action of $\U(n)$ on several copies of $\CC^n$}\label{Unaction}\label{realss}
Let $V$ and $W$ be vector spaces over $\CC$. Suppose that $dim_\CC(V)=n$ and $dim_\CC(W)=d$. We look for a Poisson Lie group covariant Poisson bracket on the space
\be
\ba
M &= Hom(W,V)\\
&=
V\otimes W^*\\
&=
Mat_{n\times d}(\CC),
\ea
\ee
treated as real, so that all functions are real-valued.

As for the previous example, it is convenient to treat the actions of the groups $K_n=\U(n)$ and $K_d=\U(d)$ on $M$,
\be
k\cdot v = k v,\quad\hbox{and}\quad \tilde k\cdot v = v\tilde k^{-1},\quad k\in K_n,\ \tilde k\in K_d,
\ee
by extending to actions of the doubles, and then restricting to $K_n$ and $K_d$ as Poisson subgroups. This is the Lu-Weinstein example. We may apply the results and notation of subsection \ref{LuWeinstein}, adding labels $n$ and $d$ where appropriate, though we'll drop them when they are implicitly obvious.
We consider the action of $G_n$ and of $G_d$ on $M$
\be
g\cdot v = gv,\quad\hbox{and}\quad h\cdot v = vh^{-1},\quad g\in G_n,\ h\in G_d.
\ee
Thus we get maps from the Lie algebras $\g_n=Lie(G_n)$ and $\g_d=Lie(G_d)$ to $Vect(M)$:
\be
\ba
\g_n\owns X\mapsto X_M\in Vect(M)\ &:\ T_vM\owns X_M(v) = Xv,\\
\g_d\owns A\mapsto A_M\in Vect(M)\ &:\ T_vM\owns A_M(v) = -vA.
\ea
\ee
Making use of the non-degenerate inner products on $\g_n$ and $\g_d$ given by the imaginary part of the trace forms,
\be
\ba
\langle X_1,X_2\rangle &= Im\,\tr(X_1X_2)\quad X_1,X_2\in\g_n,\\
\langle Y_1,Y_2\rangle &= Im\,\tr(Y_1Y_2)\quad Y_1,Y_2\in\g_d,
\ea
\ee
the r-matrices on $\g_n$ and $\g_d$ are defined by the respective decompositions $\g=\k\oplus\b$:
\be
\ba
&\g_n\wedge\g_n\owns R \sim P_{\k_n} - P_{\b_n}\\
&\g_d\wedge\g_d\owns \rho \sim P_{\k_d} - P_{\b_d},
\ea
\ee
where $\k_n=Lie(K_n)$, $\b_n=Lie(B_n)$, $\k_d=Lie(K_d)$ and $\b_d=Lie(B_d)$. As $M$ is a linear space, the cotangent space to $M$ can be identified with $M$, thus
\be
\ba
T_v^*M &= Hom(V,W)\\
&=
W\otimes V^*\\
&=
Mat_{d\times n}(\CC)
\ea
\ee
and $T_v^*M$ is spanned by the forms
\be
\xi : v\mapsto Im\,\tr(\xi^\dag v), \qquad \xi\in Mat_{n\times d}(\CC).
\ee
For any $X,Y\in\g_n$, the bivector $(X\wedge Y)_M$ on $M$ gives
\be
\ba
(X\wedge Y)_M(v)(\xi,\eta) &= Im\,\tr\bigl(\xi^\dag Xv\bigr) Im\,\tr\bigl(\eta^\dag Yv\bigr) - Im\,\tr\bigl(\xi^\dag Yv\bigr) Im\,\tr\bigl(\eta^\dag Xv\bigr)\\
&=
\langle X\wedge Y, v\xi^\dag\otimes v\eta^\dag\rangle,
\ea
\ee
from which we deduce
\be
\ba
R_M(v)(\xi,\eta) &= \langle R, v\xi^\dag\otimes v\eta^\dag\rangle \\
&=
Im\,\tr\Bigl( v\xi^\dag\Bigl( (v\eta^\dag)_{\k_n} - (v\eta^\dag)_{\b_n}\Bigr)\Bigr)\\
&=
Im\,tr \Bigl(v\xi^\dag\bigl( 2(v\eta^\dag)_{\k_n} - v\eta^\dag\bigr)\Bigr).
\ea
\ee
In a similar fashion, we obtain
\be
\rho_M(v)(\xi,\eta) = Im\,\tr\Bigl(\xi^\dag v\bigl(2(\eta^\dag v)_{\k_d} - \eta^\dag v\bigr)\Bigr).
\ee
Analogously to Proposition \ref{Rrhobrackets} we easily compute the Schouten brackets,
\[
\ba
{}[R_M,\rho_M]&=0,\\
[R_M,R_M](v)(\xi,\eta,\zeta) &= Im\,\tr\Bigl(v\xi^\dag [ v\eta^\dag , v\zeta^\dag ]\Bigr),\\
[\rho_M,\rho_M](v)(\xi,\eta,\zeta) &= -Im\,\tr\Bigl( \xi^\dag v [ \eta^\dag v , \zeta^\dag v ]\Bigr).
\ea
\]
As in the previous example we introduce the invariant Poisson structure
\be
\Pi(v)(\xi,\eta) = Im\,\tr\bigl( \xi^\dag\eta\bigr),
\ee
and putting these together, we can formulate
\begin{proposition}
For any $\lambda\in\RR$, the bi-vector
\[
P_\lambda =  R_M + \rho_M - 2\lambda\Pi
\]
is a Poisson tensor on $M$ and the actions of $K_n$ and of $K_d$ on $M$ are Poisson actions when $K_n$ and $K_d$ have the respective Poisson brackets given by \eqref{KPb}.
For $\lambda=0$, the actions of $G_n$ and of $G_d$ on $M$ are Poisson actions when $G_n$ and $G_d$ have the respective Sklyanin brackets defined by $R$ and $\rho$. 
\end{proposition}

\begin{proof}
The Poisson property $[P_\lambda , P_\lambda ]=0$ is checked by a computation analogous to that in Proposition \ref{Rrhobrackets}, in which the right hand sides to $[R_M,R_M]$ and $[\rho_M,\rho_M]$ again cancel each other, whilst invariance of the Poisson structure $\Pi$ gives it a zero Schouten bracket with $R_M$ and with $\rho_M$.
For zero $\lambda$, the last part of the proposition is obvious. For nonzero $\lambda$ the Poisson property applies to the Poisson subgroups $K_n\subset G_n$ and $K_d\subset G_d$ due to the invariance of $\Pi$ with respect to the actions of $K_n$ and of $K_d$.
\end{proof}

It is convenient to use $(\half)$ times $P$ to define the Poisson bracket.

\begin{proposition}\label{ZakPBreal}
The bracket $\{\ddarg\}$ on $M$, defined for linear functions $F,H\in\cinf(M,\RR)$
\[
F(v) = Im\,\tr(\xi^\dag v),\quad H(v) = Im\,\tr(\eta^\dag v),
\]
by
\[
\{F,H\}(v) = Im\,\tr\Bigl( v\xi^\dag(v\eta^\dag)_{\k_n} + \xi^\dag v (\eta^\dag v)_{\k_d} - \xi^\dag v\eta^\dag v - \lambda\xi^\dag\eta\Bigr) ,
\]
defines a Poisson bracket on $M$. Equivalently, the Hamiltonian vector field corresponding to $F$ is
\[
\XX_F(v) =
(v\xi^\dag)_{\k_n}v + v(\xi^\dag v)_{\k_d} - v\xi^\dag v -\lambda\xi .
\]
The actions of $K_n$ and $K_d$ on $M$, given by
\[
g\cdot v = gv,\quad h\cdot v = vh^{-1}
\]
are Poisson actions when $K_n$ and $K_d$ are endowed with Poisson brackets of the form \eqref{KPb}.
\end{proposition}

If $\lambda\neq0$, we define the maps $\Phi_\lambda^{(n)} : M \to \hermp_n$ and $\Phi_\lambda^{(d)} : M\to \hermp_d$, by
\be\label{realmomGH}
\Phi_\lambda^{(n)}(v) = \lambda Id_n + vv^\dag,\qquad \Phi_\lambda^{(d)}(v) = [\lambda Id_d + v^\dag v]^{-1},
\ee
interpreted as maps to $B_n$ and $B_d$ via the identification \eqref{identifyBP}.

\begin{proposition}
For $\lambda\neq0$, the maps $\Phi_\lambda^{(n)}$ and $\Phi_\lambda^{(d)}$ are momentum maps for the Poisson actions on $M$ of $K_n$ and $K_d$ respectively.
\end{proposition}

\begin{proof}
Consider, for $F\in\cinf(M,\RR)$ with $\delta_vF=\xi^\dag$,
\[
\ba
\left( d\Phi_\lambda^{(n)}\cdot\XX_F\right) &(v)\\
&=
\bigl((v\xi^\dag)_{\k_n}v + v(\xi^\dag v)_{\k_d} - v\xi^\dag v -\lambda\xi\bigr) v^\dag
+
v\bigl((v\xi^\dag)_{\k_n}v + v(\xi^\dag v)_{\k_d} - v\xi^\dag v -\lambda\xi\bigr)^\dag\\
&=
-\lambda(\xi v^\dag + v\xi^\dag) + \bigl((v\xi^\dag)_{\k_n} - v\xi^\dag\bigr)vv^\dag + vv^\dag\bigl((v\xi^\dag)_{\k_n} - v\xi^\dag\bigr)^\dag\\
&=
-\bigl(v\xi^\dag\bigr)_{\b_n}[\lambda + vv^\dag] - [\lambda +vv^\dag] \bigl(v\xi^\dag\bigr)_{\b_n}{}^\dag,
\ea
\]
where the last step follows from the decomposition $v\xi^\dag = (v\xi^\dag)_{\k_n} + (v\xi^\dag)_{\b_n}$; thus,
\[
v\xi^\dag + \xi v^\dag = v\xi^\dag + (v\xi^\dag)^\dag = (v\xi^\dag)_{\b_n} + (v\xi^\dag)_{\b_n}{}^\dag .
\]
Given that $\Phi_\lambda^{(n)}(v)$ represents an element $b$ in $B_n$ via $bb^\dag = \lambda +vv^\dag$, we have
\[
\left( d\Phi_\lambda^{(n)}\cdot\XX_F\right) (v) \sim \exp(-t(v\xi^\dag)_{\b_n})\cdot\Phi_\lambda^{(n)}(v),
\]
so that, for $X\in\k_n$,
\[
\ba
\left\langle\tilde X, \left( d\Phi_\lambda^{(n)}\cdot\XX_F\right) (v)\right\rangle &=
-\langle X, (v\xi^\dag)_{\b_n}\rangle\\
&=
-Im\,\tr\bigl(\xi^\dag Xv\bigr)\\
&=
-\dt F\bigl(\exp(tX)\cdot v\bigr)\\
&=
-\langle X_M,dF\rangle(v);
\ea
\]
that is, for any $F\in\cinf(M,\RR)$,
\[
\langle\tilde X, ( d\Phi_\lambda^{(n)}\cdot\XX_F )\rangle = - \langle X_M,dF\rangle.
\]

Consider
\[
\ba
\bigl( d&\Phi_\lambda^{(d)}\cdot  \XX_F\bigr) (v)\\
&=
-[\lambda+v^\dag v]^{-1}{\Bigl(}v^\dag \bigl((v\xi^\dag)_{\k_n}v + v(\xi^\dag v)_{\k_d} - v\xi^\dag v -\lambda\xi\bigr)
+
\bigl((v\xi^\dag)_{\k_n}v + v(\xi^\dag v)_{\k_d} - v\xi^\dag v -\lambda\xi\bigr)^\dag v \Bigr)[\lambda+v^\dag v]^{-1}\\
&=
-[\lambda+v^\dag v]^{-1}\left(-\lambda(v^\dag\xi + \xi^\dag v) + v^\dag v\bigl((\xi^\dag v)_{\k_d} - \xi^\dag v\bigr)+ \bigl((v\xi^\dag)_{\k_d} - \xi^\dag v\bigr)^\dag v^\dag v\right)[\lambda+v^\dag v]^{-1}\\
&=
-[\lambda+v^\dag v]^{-1}\left(-\bigl(\xi^\dag v\bigr)_{\b_d}{}^\dag[\lambda + v^\dag v] - [\lambda +v^\dag v] \bigl(\xi^\dag v\bigr)_{\b_d}\right)[\lambda+v^\dag v]^{-1}\\
&=
(\xi^\dag v)_{\b_d}[\lambda+v^\dag v]^{-1} + [\lambda+v^\dag v]^{-1}\bigl((\xi^\dag v)_{\b_d}\bigr)^\dag.
\ea
\]
Given that $\Phi_\lambda^{(d)}(v)$ represents an element $b$ in $B_d$ via $\hermp\owns bb^\dag = \lambda+v^\dag v$, we have
\[
\left( d\Phi_\lambda^{(d)}\cdot\XX_F\right) (v) \sim \exp(t(\xi^\dag v)_{\b_d})\cdot\Phi_\lambda^{(d)}(v),
\]
so that, for $X\in\k_d$,
\[
\ba
\left\langle\tilde X, \left( d\Phi_\lambda^{(d)}\cdot\XX_F\right) (v)\right\rangle &=
\langle X, (\xi^\dag v)_{\b_d}\rangle\\
&=
Im\,\tr\bigl(\xi^\dag vX\bigr)\\
&=
\dt F\bigl(\exp(-tX)\cdot v\bigr)\\
&=
-\langle X_M,dF\rangle(v);
\ea
\]
that is, for any $F\in\cinf(M,\RR)$,
\[
\langle\tilde X, ( d\Phi_\lambda^{(d)}\cdot\XX_F )\rangle =  -\langle X_M,dF\rangle.
\]

\end{proof}

\begin{proposition}\label{actderivsreal}
Let $f\in\cinf(M,\RR)$. The action-derivatives $D_vf\in\b_n$ and $\nabla_vf\in\k_n$, defined in \eqref{actionderivatives}, are 
\[
D_vf = (v\delta_vf)_\b,\qquad  \nabla_vf = (v\delta_vf)_\k.
\]
\end{proposition}

\begin{proof}
It is convenient to suppose, without loss of generality, that $f$ is the linear function,
\[
f(v) = Im\,\tr(\xi^\dag v),
\]
for which
\[
\delta_vf=\xi^\dag.
\]
We compute
\[
\dt f(e^{t\sigma}v) = Im\,\tr(\xi^\dag\sigma v).
\]
For $\sigma\in\k_n$, this gives
\[
Im\,\tr(D_vf \sigma) = \langle D_vf,\sigma\rangle = Im\,\tr(\xi^\dag\sigma v) = \langle(v\xi^\dag)_\b,\sigma\rangle
\]
and for $\sigma\in\b_n$, we have
\[
Im\,\tr(\nabla_vf \sigma) = \langle \sigma,\nabla_vf\rangle = Im\,\tr(\xi^\dag\sigma v) = \langle\sigma,(v\xi^\dag)_\k\rangle.
\]
Hence the result.
\end{proof}



\section{Examples of semidirect products}

In this section, using the formulae of Section \ref{semidirectsec}, we will look at the semidirect product of $M$ and $N$ for the examples of the previous section.

\subsection{The real case: action of $\U(n)$ on several copies of $\CC^n$}
With respect to subsection \ref{realss}, suppose we have the Poisson spaces 
 \be
(M,P_M) =
\Bigl(\bigl\{ Mat_{n\times d_1}(\CC)\bigr\}\,,\, R_M + \rho^{(1)}_M + \lambda\Pi_M\Bigr),
\ee
 
\be
(N,P_N) =
\Bigl(\bigl\{Mat_{n\times d_2}(\CC)\bigr\}\,,\,  R_N + \rho^{(2)}_N + \mu\Pi_N\Bigr),
\ee
with $\rho^{(1)}$ an r-matrix on $gl(d_1,\CC)$, with $\rho^{(2)}$ an r-matrix on $gl(d_2,\CC)$, both skew-symmetric with respect to the respective forms $Im\,\tr$, and with $\lambda,\mu\in\RR$ possibly different. Consider the map $\S:M\times N\to \tilde M=\CC^{n\times(d_1+d_2)}$,
\be
\S(v;V) = \bigl(v , \Phi^G_\lambda(v)\cdot V\bigr) = (v,b(v)V)
\ee
with $\Phi^G_\lambda(v) = \lambda+vv^\dag$ from \eqref{realmomGH}, and $b$ defined by $\Phi^G_\lambda(v) = b(v)b(v)^\dag$. Define the linear functions on $\tilde M$,
\be
\ba
\tilde F(\tilde V) = \tr\bigl(\tilde X^\dag\tilde V ),\qquad \tilde H(\tilde V) = \tr\bigl(\tilde Y^\dag\tilde V ).
\ea
\ee
Our aim is to prove
\begin{proposition}\label{concatenatedPBreal}
For $\mu=1$, the Poisson bracket on $\tilde M$ engendered by $\S$---that is, such that $\S:M\times N\to\tilde M$ is a Poisson isomorphism---is given by
\[
\ba
\{\tilde F,\tilde H\} (\tilde V) &=
\half Im\,\tr\bigl( \tilde V \tilde X^\dag (\tilde V \tilde Y^\dag)_{\k_n} +  \tilde X^\dag \tilde V (\tilde Y^\dag \tilde V)_{\k_{(d_1+d_2)}} - \tilde V\tilde X^\dag\tilde V\tilde Y^\dag - \lambda \tilde X^\dag \tilde Y\bigr) .
\ea
\]
\end{proposition}
The proof of this proposition is essentially a direct application of Proposition \ref{semidirect}. We may split it into several parts.

In the notation of Proposition \ref{semidirect}, we introduce the linear functions
\[
\ba
&f(v) = Im\,\tr(\xi^\dag v) , \quad h(v) = Im\,\tr(\eta^\dag v),\\
&\varphi(V) = Im\,\tr(X^\dag V) , \quad \psi(V) = Im\,\tr(Y^\dag V).
\ea
\]

Let's denote the Poisson bracket on $N$ $\{\ ,\ \}_0 + \mu\{\ ,\ \}_\Pi$. 
Define the functions
\be
\ba
F(v ; V) &= f(v) + \varphi(V) , \\
H(v ; V) &= h(v) + \psi(V)  .
\ea
\ee
Thus, we have
\[
\S^*F(v ; V) = \tr\bigl(\xi^\dag v + X^\dag bV),
\]
with $\Phi^G_\lambda(v) = bb^\dag$.
As in the proof of Proposition \ref{semidirect}, we may find the contribution of $\Pi$ to the modified Poisson structure by computing
\be\label{consttermreal}
\ba
\{\S^*F,\S^*H\}_\Pi (v ; V ) &= Im\,\tr \bigl( (X^\dag b) (b^\dag Y) \bigr)\\
&=
Im\,\tr ( X^\dag bb^\dag Y ) \\
&=
Im\, \tr(X^\dag[\lambda+vv^\dag]Y).
\ea
\ee 
The following Lemma is a direct consequence of Proposition \ref{actderivsreal}.
\begin{lemma}\label{actderivstermreal}
The term $\langle D_vf , \nabla_V\psi\rangle$ from Proposition \ref{semidirect} is
\[
\langle D_vf , \nabla_V\psi\rangle = 
Im\,\tr\bigl( v\xi^\dag (VY^\dag)_\k \bigr).
\]
\end{lemma}

\begin{proofof}{Proposition}{concatenatedPBreal} 
Implementing 
Proposition \ref{semidirect}, making use of \eqref{consttermreal} and Lemma \ref{actderivstermreal}, then rearranging terms, we have\footnote{A cavalier approach is taken to the use of `$\tr$' in that the same operation impinges on matrices of differing sizes. This just avoids having to write out lots of different terms. The meaning should be clear.} 
\[
\ba
2\{F,H\}&(v;V)\\ 
&= \{f,h\}_M(v) + \{\varphi,\psi\}_0(V) + \langle D_vf,\nabla_V\psi\rangle - \langle D_vh,\nabla_V\varphi\rangle
- \mu Im\,\tr(X^\dag[\lambda+vv^\dag]Y) \\
&=\\
&\quad
Im\,\tr\bigl( (v\xi^\dag + VX^\dag)(v\eta^\dag + VY^\dag)_{\k_n} + \xi^\dag v (\eta^\dag v)_{\k_{d_1}} + X^\dag V (Y^\dag V)_{\k_{d_2}} \\
&\qquad\qquad - v\xi^\dag v\eta^\dag - VX^\dag VY^\dag - VX^\dag v\eta^\dag - \mu X^\dag vv^\dag Y
-
\lambda(\xi^\dag\eta + \mu X^\dag Y)\bigr).
\ea
\]
Now, let us write
\[
\tilde V = (v,V),\ \ \tilde X = (\xi,X),\ \ \tilde Y = (\eta,Y),
\]
so that
\[
\tilde V \tilde X^\dag = (v,V)\left(\begin{array}{c} \xi^\dag \\ X^\dag\end{array}\right)
=
v\xi^\dag + VX^\dag,
\]

\[
\ba
\tilde V\tilde X^\dag\tilde V\tilde Y^\dag = (v\xi^\dag + VX^\dag)(v\eta^\dag + VY^\dag)
\ea
\]

\[
\tilde X^\dag \tilde V =
\left(\begin{array}{cc} \xi^\dag v & \xi^\dag V \\ X^\dag v & X^\dag V \end{array}\right)
=
\left(\begin{array}{cc} (\xi^\dag v)_{k_1} & - v^\dag X \\ X^\dag v & (X^\dag V)_{\k_2} \end{array}\right)
+ 
\left(\begin{array}{cc} (\xi^\dag v)_{\b_1} & \xi^\dag V + v^\dag X \\ 0 & (X^\dag V)_{\b_2} \end{array}\right);
\]
that is,
\[
(\tilde X^\dag\tilde V)_{\k_{(d_1+d_2)}} = 
\left(\begin{array}{cc} (\xi^\dag v)_{k_{d_1}} & - v^\dag X \\ X^\dag v & (X^\dag V)_{\k_{d_2}} \end{array}\right).
\]

Using these expressions, it is straightforward to verify Proposition \ref{concatenatedPBreal}.

\end{proofof}

\subsection{The holomorphic case: action of $\GL(n,\CC)$ on several copies of $\CC^n\times\CC^n$}
Similarities between this subsection and the previous one allows us to curtail the proofs, which would otherwise be unreasonably long. For any reader who finds an inadequate supply of details in what follows, it should be sufficient to make a comparison with analogous proofs for the real case.

With reference to subsection \ref{holss}, suppose we have the Poisson spaces 
 \be
M =
\Bigl(\Bigl\{ m=(v,w^T)\,\left|\, v,w\in Mat_{n\times d_1}(\CC)\Bigr\}\,,\, \{\ ,\ \}\sim \hat R_M + \hat\rho^{(1)}_M + \lambda\Pi_M\Bigr)\, ,\right.
\ee
 
\be
N =
\Bigl(\Bigl\{ n=(V,W^T)\, \left|\, V,W\in Mat_{n\times d_2}(\CC)\Bigr\}\,,\, \{\ ,\ \}\sim \hat R_N + \hat\rho^{(2)}_N + \mu\Pi_N\Bigr)\, ,\right.
\ee
with $\rho^{(1)}$ an r-matrix on $gl(d_1,\CC)$, with $\rho^{(2)}$ an r-matrix on $gl(d_2,\CC)$, both skew-symmetric with respect to the respective trace-forms, and with $\lambda,\mu\in\CC$ possibly different. Consider the map $\S:M\times N\to  M\times N$,
\[
\S(v,w^T;V,W^T) = \bigl(v,w^T; \Phi^G_\lambda(v,w^T)\cdot(V,W^T)\bigr),
\] 
with $\Phi^G_\lambda = \lambda + vw^T$ from \eqref{holmomGH}, and define the concatenation map $\cC: M\times N\to \tilde M=\CC^{n\times(d_1+d_2)}\times\CC^{(d_1+d_2)\times n}$,
\be
\cC(v,w^T;V,W^T) = (\tilde V,\tilde W^T),\quad\hbox{ with } \tilde V= (v,V),\ \tilde W = (w,W). 
\ee
Define the linear functions on $\tilde M$,
\be
\ba
\tilde F(\tilde V,\tilde W^T) &= \tr\bigl(\tilde X^T\tilde V + \tilde W^T\tilde P\bigr),\\
\tilde H(\tilde V,\tilde W^T) &= \tr\bigl(\tilde Y^T\tilde V + \tilde W^T\tilde Q\bigr).
\ea
\ee
Our aim is to prove
\begin{proposition}\label{concatenatedPBhol}
For $\mu=1$, the Poisson bracket on $\tilde M$ engendered by the composition $\cC\circ\S$ is given by
\[
\ba
\{\tilde F,\tilde H\} (\tilde V,\tilde W^T) &=
\half\tr\bigl( - (\tilde V \tilde X^T - \tilde P \tilde W^T) R (\tilde V \tilde Y^T - \tilde Q \tilde W^T) - (\tilde X^T \tilde V - \tilde W^T \tilde P)\tilde\rho(\tilde Y^T \tilde V - \tilde W^T \tilde Q) \\
&\qquad+
\tilde W^T \tilde V(\tilde X^T \tilde Q - \tilde Y^T \tilde P) + \tilde V \tilde W^T(\tilde Q \tilde X^T - \tilde P \tilde Y^T) + 2\lambda (\tilde X^T \tilde Q - \tilde Y^T \tilde P)\bigr),
\ea
\]
where the r-matrix $\tilde\rho$ on $gl(d_1+d_2,\CC)$ is given by
\[
\tilde\rho\left(\begin{array}{cc}A_{11} &A_{12}\\ A_{21} &A_{22}\end{array}\right)
=
\left(\begin{array}{cc}\rho_1(A_{11}) & A_{12} \\ - A_{21} & \rho_2(A_{22})\end{array}\right),
\] 
for $A_{11}\in Mat(d_1\times d_1,\CC)$, $A_{12}\in Mat(d_1\times d_2,\CC)$, $A_{21}\in Mat(d_2\times d_1,\CC)$, $A_{22}\in Mat(d_2\times d_2,\CC)$.
\end{proposition}
As for the real example, this proposition is a direct application of Proposition \ref{semidirect}. As before, we split it into several parts.

In the notation of Proposition \ref{semidirect}, let us introduce the linear functions
\[
\ba
&f(v,w^T) = \tr(x^Tv+w^T\xi),\quad h(v,w^T) = \tr(y^Tv+w^T\eta),\\
&\varphi(V,W^T) = \tr(X^TV+W^TP),\quad \psi(V,W^T) = \tr(Y^TV+W^TQ).
\ea
\]

Let's denote the Poisson bracket on $N$ $\{\ ,\ \}_0 + \mu\{\ ,\ \}_\Pi$. 
Define the functions
\be
\ba
F(v,w^T; V,W^T) &= f(v,w^T) + \varphi(V,W^T) , \\
H(v,w^T; V,W^T) &= h(v,w^T) + \psi(V,W^T) .
\ea
\ee
Thus, writing $\Phi^G_\lambda(v,w^T) = g_+g_-^{-1}=g$ for short, we have
\[
\S^*F(v,w^T; V,W^T) = \tr\bigl(x^Tv + w^T\xi + X^Tg_+V + W^Tg_-^{-1}P\bigr).
\]
As in the proof of Proposition \ref{semidirect}, we find the contribution of $\Pi$ to the modified Poisson structure by computing
\be\label{consttermhol}
\ba
\{\S^*F,\S^*H\}_\Pi (v,w^T;V,W^T) &= \tr \bigl( (X^Tg_+) (g_-^{-1} Q) - (Y^Tg_+)(g_-^{-1}P)\bigr)\\
&=
\tr ( X^TgQ - Y^TgP) \\
&=
\tr(X^T[\lambda+vw^T]Q - Y^T[\lambda+vw^T]P).
\ea
\ee
Next, we need

\begin{lemma}\label{actderivstermhol}
The term $\langle D_{(v,w^T)}f , \nabla_{(V,W^T)}\psi\rangle$ from Proposition \ref{semidirect} is
\[
\ba
\langle D_{(v,w^T)}f , \nabla_{(V,W^T)}\psi\rangle &= 
\half\tr\bigl( - (vx^T-\xi w^T) R (VY^T - QW^T)\bigr) \\ 
&\qquad
+ \half\tr\bigl( vx^TVY^T - \xi w^T QW^T + W^Tvx^TQ - w^TVY^T\xi\bigr)\,.
\ea
\]
\end{lemma}

\begin{proof}
Substituting from Proposition \ref{actionderivshol} we have
\[
\ba
\langle D_{(v,w^T)}f , \nabla_{(V,W^T)}\psi\rangle = 
\half\tr\Bigl((vx^T-\xi w^T)\bigl( (VY^T + QW^T) - R(VY^T - QW^T)\bigr)\Bigr)\, ,
\ea
\]
and the result follows after some rearrangement of terms.
\end{proof}

\begin{proofof}{Proposition}{concatenatedPBhol} 
Implementing 
Proposition \ref{semidirect}, making use of \eqref{consttermhol} and Lemma \ref{actderivstermhol}, then rearranging terms, we have
\[
\ba
2\{F,H\}&(v,w^T;V,W^T)\\ 
&= \{f,h\}_M(v,w^T) + \{\varphi,\psi\}_0(V,W^T) + \langle D_{(v,w^T)}f,\nabla_{(V,W^T)}\psi\rangle - \langle D_{(v,w^T)}h,\nabla_{(V,W^T)}\varphi\rangle \\
&\qquad\qquad\qquad\qquad+ 2\mu\tr(X^T[\lambda+vw^T]Q - Y^T[\lambda+vw^T]P).  \\
\ea
\]
Now, let us write
\[
\tilde V = (v,V),\ \  \tilde W = (w,W),\ \ \tilde X = (x,X),\ \ \tilde Y = (y,Y),\ \ \tilde P = (\xi,P),\ \ \tilde Q = (\eta,Q),
\]
so that
\[
\tilde V \tilde X^T - \tilde P \tilde W^T = (v,V)\left(\begin{array}{c} x^T \\ X^T\end{array}\right)- (\xi,P) \left(\begin{array}{c} w^T \\ W^T\end{array}\right)
=
vx^T - \xi w^T + VX^T - PW^T,
\]

\[
\tilde X^T \tilde V - \tilde W^T \tilde P = 
\left(\begin{array}{cc} x^Tv - w^T\xi & x^TV - w^TP \\ X^Tv - W^T\xi & X^TV - W^TP\end{array}\right),
\]

\[
\ba
\tilde V \tilde W^T(\tilde Q\tilde X^T - \tilde P\tilde Y^T)
&=
(v,V)\left(\begin{array}{c}w^T \\ W^T\end{array}\right)\left( (\eta,Q)\left(\begin{array}{c} x^T \\ X^T\end{array}\right) - (\xi,P) \left(\begin{array}{c} y^T \\ Y^T\end{array}\right) \right) \\
&=
(vw^T+VW^T)(\eta x^T - \xi y^T + WX^T - PY^T),
\ea
\]

\[
\ba
\tilde W^T\tilde V ( \tilde X^T\tilde Q - \tilde Y^T\tilde P) 
&= 
\left(\begin{array}{c}w^T \\ W^T\end{array}\right) (v,V) \left( \left(\begin{array}{c} x^T \\ X^T\end{array}\right) (\eta,Q) -  \left(\begin{array}{c} y^T \\ Y^T\end{array}\right) (\xi,P) \right)\\
&=
\left(\begin{array}{cc} w^Tv & w^TV \\ W^Tv & W^TV\end{array}\right)
\left(\begin{array}{cc} x^T\eta - y^T\xi & x^TQ - y^TP \\ X^T\eta - Y^T\xi & X^TQ - Y^TP\end{array}\right).
\ea
\]
Using these expressions, it is straightforward to verify Proposition \ref{concatenatedPBhol}.

\end{proofof}



\section{Applications: half-dressing}

As mentioned earlier, half-dressing arose quite naturally in a study of spin Ruijsenaars models in the setting of Hamiltonian reduction. This section treats half-dressing as a special case of the semidirect product of Poisson $G$-spaces.\footnote{The nomenclature \textit{half-dressing} has been chosen in preference to the more natural \textit{dressing} simply for the sake of consistency with respect to \cite{FFM}.}

\subsection{Half-dressing on $\CC^n$ for the $\U(n)$-action}
Let $v$ and $w$ be vectors in $\CC^n$. From $v$, let us form the positive-definite, hermitian matrix 
\be
\Phi(v) = 1+vv^\dag.
\ee
Any positive definite, hermitian matrix may be decomposed as the product of an element in $B$ and its hermitian conjugate; that is, there exists a unique $b\in B$ such that $bb^\dag = \Phi$. We obtain therefore a map from $\CC^n$ to $B$:
\be
v\mapsto b(v) : b(v) b(v)^\dag = \Phi(v) = 1 + vv^\dag.
\ee
The half-dressing of $w$ by $v$ is the map
\be
v\cdot w = b(v)w
\ee
in which $b(v)w$ is usual matrix multiplication. What is done in the spin Ruijsenaars model is to take several vectors like these and to combine them using sequential half-dressing. Consider the ordered collection $v_1,v_2,\dots, v_d\in\CC^n$, and construct from this collection a \textit{half-dressed}, similarly ordered, collection $\tilde v_1, \tilde v_2, \dots,\tilde v_d$ by the following procedure:
\be
\ba
\tilde v_1 &= v_1,\\
\tilde v_2 &= v_1\cdot v_2 = \tilde v_1\cdot v_2,\\
\tilde v_3 &= v_1\cdot(v_2\cdot v_3) = \tilde v_2\cdot v_3,\\
\vdots\ &\\
\tilde v_d &= v_1\cdot(v_2\cdot(v_3\cdot(\dots \cdot(v_{d-1}\cdot v_d)\dots))) = \tilde v_{d-1}\cdot v_d.
\ea
\ee
It was necessary to look at the multiple copies of $\CC^n$ as carrying a Poisson structure, and this Poisson structure has the natural property of making $\CC^n$ into a Poisson $G$-space, with $G=\U(n)$. It is clear that this mapping from $(v_1,\dots,v_d)$ to $(\tilde v_1,\dots,\tilde v_d)$ is invertible, and so one may ask what the Poisson structure looks like in the new variables; indeed this played an intrinsic role in \cite{FFM}.

With insight gained from the notion of semidirect product of Poisson $G$-spaces, this may be seen as a straightforward application of the example of the last section: the case $d_1=1$ of Proposition \ref{concatenatedPBreal} provides us with the inductive step for the whole construction.
The details of this induction argument should be clear, and are left to any interested reader to work through for themselves. 

\subsection{Half-dressing on $\CC^n\times\CC^n$ for the $\GL(n,\CC)$-action}
It is worth some attention to formulate the analogue of half-dressing in the holomorphic 
setting, where it is easily adapted from the notion of semidirect product.

Let $(v_1,w_1^T)$ and $(v_2,w_2^T)$ be two pairs of vectors in $\CC^n$. From the first pair, let us form the invertible matrix
\be
\Phi(v_1,w_1^T) = 1 + v_1w_1^T.
\ee
We decompose this matrix as the product of an upper-triangular matrix $\Phi_+$ and a lower-triangular matrix $\Phi_-^{-1}$, with the property that the diagonal parts of both matrices are the same. That is to say, as mentioned just after \eqref{factH}, we restrict to a suitable open subset in $\CC^n\times\CC^n$ in which $\Phi$ is invertible, and factorizable. We obtain in this way a map\footnote{As mentioned in subsection \ref{holex}, the map is defined modulo multiplication on the right by a pair of matrices of the form $(\Delta,\Delta)$, where $\Delta$ is a diagonal matrix with entries in $\{1,-1\}$. Hence, in reality we have $2^n$ maps. Any one of them will do, and the claims which follow are still valid.}  from $\CC^n\times\CC^n$ to $HN_+\times HN_-$:
\be
(v,w^T)\mapsto (\Phi(v,w^T)_+\,,\,\Phi(v,w^T)_-) \ : \ \Phi(v,w^T)_+ \Phi(v,w^T)_-^{-1} = \Phi(v,w^T) = 1+vw^T.
\ee
The half-dressing of $(v_2,w_2^T)$ by $(v_1,w_1^T)$ is the map
\be
(v_1,w_1^T)\cdot (v_2,w_2^T) = \bigl(\Phi(v_1,w_1^T)_+v_2\,,\, w_2^T\Phi(v_1,w_1^T)_-^{-1}\bigr) .
\ee
Now we take an ordered collection $(v_1,w_1^T), (v_2,w_2^T),\dots,(v_d,w_d^T)\in\CC^n\times\CC^n$ of pairs of vectors in $\CC^n$, and construct the \textit{half-dressed}, similarly ordered, collection $(\tilde v_1, \tilde w_1^T), (\tilde v_2, \tilde w_2^T),\dots,(\tilde v_d, \tilde w_d^T)$ by the following procedure:
\be\label{seqdress}
\ba
(\tilde v_1,\tilde w_1^T) &= (v_1,w_1^T),\\
(\tilde v_2,\tilde w_2^T) &= (v_1,w_1^T)\cdot(v_2,w_2^T) = (\tilde v_1,\tilde w_1^T) \cdot (v_2,w_2^T) ,\\
(\tilde v_3,\tilde w_3^T) &= (v_1,w_1^T)\cdot ((v_2,w_2^T)\cdot (v_3,w_3^T))
= (\tilde v_2,\tilde w_2^T) \cdot (v_3,w_3^T) ,\\
\vdots\qquad&\\
(\tilde v_d,\tilde w_d^T) &= (v_1,w_1^T)\cdot ((v_2,w_2^T)\cdot ((v_3,w_3^T)\cdot (\cdots (v_{d-1},w_{d-1}^T)\cdot(v_d,w_d^T)\dots)))
= (\tilde v_{d-1},\tilde w_{d-1}^T) \cdot (v_d,w_d^T) .
\ea
\ee
We will prove the following
\begin{proposition}
Suppose that the Poisson bracket on $(\CC^n\times\CC^n)^d$ is the direct sum of $d$ copies of the $d=1$ version of Proposition \ref{ZakPBhol}, with $\lambda=1$, and consider the map $\mathscr{D}:(\CC^n\times \CC^n)^d\to\CC^{n\times d}\times\CC^{d\times n}$ given by
\eqref{seqdress} :
\[
\ba
&\mathscr{D}(v_1,w_1^T; v_2,w_2^T;\dots;v_d,w_d^T) = (V,W^T)\\
&\qquad\hbox{with}\\
&V=(\tilde v_1, \tilde v_2,\dots,\tilde v_d),\quad
W =(\tilde w_1,\tilde w_2,\dots,\tilde w_d). 
\ea
\]
Then the Poisson bracket on $\CC^{n\times d}\times\CC^{d\times n}$ engendered by $\mathscr{D}$ is the bracket in Proposition \ref{ZakPBhol}.
\end{proposition}

\begin{proof}
We may prove this by induction on $d$, and applying Proposition \ref{concatenatedPBhol}. For $d=2$, the claim is equivalent to the case $d_1=1=d_2$. Presupposing that the result holds for some value of $d$, the case $d_1=1$, $d_2=d$ confirms that the result holds for $d+1$.
\end{proof}

\subsection{Comments}

Whereas in Section \ref{examples} the left and right r-matrices $R$ and $\rho$ were arbitrary, the right r-matrix $\rho$ becomes uniquely generated by the half-dressing procedure. The half-dressing described here was based on the action of the group $G$, or $K_n$, on the left. Clearly the same thing could be done by means of the action of the group $H$, or $K_d$, on the right, but in that case the left r-matrix $R$ will be restricted. In fact, for the $\U(n)$ case, one can envisage a sort of Fourier transform operation from $d$ uncoupled copies of the Poisson space $\CC^n$ to $n$ uncoupled copies of the Poisson space $\CC^d$, or the analogous operation for the $\GL(n,\CC)$ case involving $\CC^n\times\CC^n$ and $\CC^d\times\CC^d$.

In the special case $d=n$, for $\rho=R$ and if $\lambda=0$, the formulae for the Poisson brackets in Proposition \ref{ZakPBhol} and in Proposition \ref{ZakPBreal} are the standard ones for the \textit{Heisenberg double}, and they can be described in terms of a symplectic structrue with explicit form due to Alekseev and Malkin \cite{AM}. In Appendix \ref{sympstruct} it is shown by means of an explicit formula for the symplectic form in terms of the non half-dressed representation, that, for $\U(n)$ symmetry, all cases are symplectic. It would be interesting to find an analogue of the Alekseev-Malkin formula, at least for the $\U(n)$ case, and ideally for arbitrary $n$ and $d$. Moreover, it would be interesting to find such a formula by means of Hamiltonian reduction; probably starting with the analogue of the Alekseev-Malkin structure for $d=n$.


\appendix

\section{Regularisation of Zakrzewski's  $\U(n)$ Poisson-symmetric Poisson structure on $\CC^n$}\label{zakreg}
According to \cite{Z} it appears that there should be an infinite family---parametrised by a function, denoted $\alpha$ below---of Poisson structures on $\CC^n$ on which the natural action of $\U(n)$ is Poisson. As we see here, this is not quite the case, and the apparent generality can be gauged away by a suitable change of variable; to be precise, $\alpha$ can be normalized to $1$ as long as it never takes the value zero.

Zakrzewski found the following Poisson bracket on $\CC^n$
\be\label{ZakPBgen}
\{Im(\xi^\dag w),Im(\eta^\dag w)\} = Im\Bigl( \xi^\dag(w\eta^\dag)_\k w - \alpha\xi^\dag\eta - \half\xi^\dag w\eta^\dag w - \beta\xi^\dag ww^\dag \eta\Bigr),
\ee
with corresponding Hamiltonian vector field
\be\label{ZakHamvfgen}
\XX_{Im(\xi^\dag w)}(w) = (w\xi^\dag)_\k w - \alpha(|w|^2)\xi - \half(\xi^\dag w)w - \beta(|w|^2)(w^\dag \xi)w,
\ee
in which $\alpha$ and $\beta$ must be related by the condition (prime denotes derivative)
\be\label{Zakcondn}
2(\alpha-t\alpha')\beta = 2\alpha\alpha' + \alpha + t\alpha',
\ee
but are otherwise arbitrary functions of $|w|^2$.

In the case $d=1$, the formulae in Proposition \ref{ZakPBreal} correspond to $\alpha=\lambda$, $\beta=\half$, and it looks like a special case. 
However, we may prove

\begin{proposition}\label{special2Zak}
The Poisson structure given by \eqref{ZakPBgen} and \eqref{ZakHamvfgen} with $\alpha=1$  can be transformed to the general $\alpha$ case by the map
\[
w\mapsto u = f(|w|^2)w,\quad\hbox{ with } f(|w|^2)=\sqrt{\alpha(|u|^2)}
\]
\end{proposition}
and
\begin{proposition}\label{Zak2special}
If $\alpha$ is never zero, the Poisson structure given by \eqref{ZakPBgen}, \eqref{ZakHamvfgen} and \eqref{Zakcondn} can be transformed to the case with $\alpha=1$ by the map
\[
w\mapsto u = f(|w|^2)w,\quad \hbox{ with }f(|w|^2)=\frac1{\sqrt{\alpha(|u|^2)}}
\]

\end{proposition}

\begin{proofof}{Proposition}{special2Zak}
Assuming that the Poisson bracket on $\CC^n$ is given by \eqref{ZakPBgen} with $\alpha=1$, $\beta=\half$, let's look at the map
\be
w\mapsto u := f(|w|^2)w,
\ee
and consider the function
\be
H(w) = Im(\xi^\dag u) = f(|w|^2) Im(\xi^\dag w).
\ee
We compute, dropping the argument $|w|^2$ of $f$ and of its derivative $f'$,
\be
d_wH = f\xi^\dag + 2\ri f' Im(\xi^\dag w)w^\dag,
\ee
and hence
\be
\XX_H(w) = f(w\xi^\dag)_\k w + 2\ri f'Im(\xi^\dag w)|w|^2w - f\xi + 2\ri f(\xi^\dag w+w^\dag\xi)w - \half(\xi^\dag w+w^\dag\xi)w,
\ee
from which we get
\[
\{|w|^2,H\} = - f(\xi^\dag w+w^\dag \xi) - |w|^2f(\xi^\dag w+w^\dag) = -(1+|w|^2)f(\xi^\dag w+w^\dag \xi).
\]
Putting these together, after some simplifications, we have
\be\label{sttoZk}
\ba
\XX_H(u) &= f' \{|w|^2,H\}w + f\XX_H(w)\\
&=
(u\xi^\dag)_\k u - f^2\xi - \half(\xi^\dag u)u - \left[\half + 2(1+|w|^2|)\frac{f'}{f}\right](u^\dag\xi)u.
\ea
\ee
Let us set
\be
\ba
f(|w|^2)^2 &= \alpha(|u|^2), \\
\half + 2(1+|w|^2|)\frac{f'}{f}(|w|^2) &= \beta(|u|^2),
\ea
\ee
so that \eqref{sttoZk} takes the general form \eqref{ZakPBgen}. 
Let $|u|^2=t$ and $|w|^2=\tau$. Then, as $|u|^2 = f(|w|^2)^2|w|^2$, we have
\be\label{tautrelns1}
\left\{\ba
t&=f(\tau)^2\tau,\\
\alpha(t)&=f(\tau)^2,\\
\beta(t) &= \half + 2(1+\tau)f'(\tau)/f(\tau).
\ea\right.
\ee
Eliminating $\tau$ from the relations \eqref{tautrelns1}, we get Zakrzewski's condition \eqref{Zakcondn} on the functions $\alpha$ and $\beta$.

\end{proofof}

\begin{proofof}{Proposition}{Zak2special}

Assuming that the Poisson bracket on $\CC^n$ is given by \eqref{ZakPBgen}, for general $\alpha$ and $\beta$ satisfying \eqref{Zakcondn},
let us make the change of variables $w\mapsto u$:
\be
u = f(\vert w\vert^2) w.
\ee
Proceeding now in an analogous fashion to that in the previous proof, we obtain
\be
\ba
\dot w = \XX_H(w) 
&=
 f\bigl[(w\xi^\dag)_\k w - \alpha\xi - \half(\xi^\dag w)w - \beta(w^\dag \xi)w\bigr]\\
&\qquad\qquad
+ \ri Im(\xi^\dag w) f'\bigl[ \vert w\vert^2 + 2\beta\vert w\vert^2 + 2\alpha\bigr]w,
\ea
\ee
which gives
\[
\ba
\{ |w|^2,H\} 
&=
- fRe(\xi^\dag w)[2\alpha+\vert w\vert^2 + 2\beta\vert w\vert^2].
\ea
\]
Hence we have
\be\label{udot}
\ba
\dot u &= f\XX_H(w) + f' \{|w|^2,H\} w\\
&=
 (u\xi^\dag)_\k u - \xi - \half(\xi^\dag u)u 
 \\
&\qquad\qquad
- \bigl[ff' \vert w\vert^2 + 2ff'\beta\vert w\vert^2 + f^2\beta+ 2ff'\alpha\bigr] (w^\dag \xi)  w.
\ea
\ee
Now, 
{on condition that the function $t\mapsto\alpha(t)$ is never zero}, we may choose $\alpha f^2=1$, so that $\alpha'f^2+2\alpha ff' =0$, or
\be
2ff' = -f^2\frac{\alpha'}{\alpha},
\ee
and subsitituting this in the square brackets expression on the second line of the equation for \eqref{udot}, 
we have, using the condition \eqref{Zakcondn},
\[
\ba
\frac{f^2}{\alpha} \bigl[ \half\alpha'\vert w\vert^2 + \alpha'\beta\vert w\vert^2 - \alpha\beta + \alpha\alpha'\bigr]
= - \frac{f^2}2.
\ea
\]
Putting all this together, we find
\[
\XX_{Im(\xi^\dag u)}(u) =  (u\xi^\dag)_\k u - \xi - \half(\xi^\dag u)u - \half (u^\dag \xi) u.
\]
\end{proofof}

\begin{remark}\label{scalingremark}
Provided that $\lambda\neq0$, the assumption $\lambda=1$ is a special case of Proposition \ref{Zak2special}.
\end{remark}



\section
{The $\U(n)$ Poisson-symmetric symplectic form on $\CC^n$}
\label{sympstruct}

To show that the Poisson structure on $\CC^n$ given by Proposition \ref{ZakPBreal} is nondegenerate, we'll pass to local coordinates to find an explicit expression for the symplectic form in these coordinates. This is then extended to an expression which is well-defined globally.

Let us introduce the notation
\be
\CC^n\owns \bk : (\bk)_l=\delta_{kl};
\ee
that is, $\bk$ is the constant vector with all components zero except for the $k$th component, which is $1$.

Let us notice that for $\xi^\dag = \bk^T$, $Im\,\tr(\xi^\dag w) = Im(w_k)$, and for $\xi^\dag =\ri\bk^T$, $Im\,\tr(\xi^\dag w)=Re(w_k)$. A convenient route to expressing the Poisson bracket in terms of the components of $w$ is to compute the Hamiltonian vector fields. The only non-trivial step involves the projection $(w\xi^\dag)_\k$. For $\xi^\dag = \ri\bl^T$,
\be
\ba
\ri w\bl^T &=  \ri \sum_{r=1}^n w_r\br\,\bl^T\\
&= 
\Bigl(\half\ri(w_l+\overline{w_l}) \bl\,\bl^T + \ri\sum_{r>l}^n[w_r\br\,\bl^T + \overline{w_r}\bl\,\br^T]\Bigr) \\
&\qquad+ 
\Bigl(\half\ri(w_l-\overline{w_l})\bl\,\bl^T + \ri\sum_{r<l} w_r\br\,\bl^T - \ri\sum_{r>l}\overline{w_r}\bl\,\br^T\Bigr)
\ea
\ee
and the first bracket is a matrix in $\k$, while the second is a matrix in $\b$. Hence, we get
\be
( \ri w\bl^T)_\k =  \Bigl( \half\ri(w_l+\overline{w_l}) \bl\,\bl^T + \ri\sum_{r>l}^n[w_r\br\,\bl^T + \overline{w_r}\bl\,\br^T]\Bigr),
\ee
whence
\be
(\ri w\bl^T)_\k w = \ri Re(w_l)w_l\bl + \ri\sum_{r>l}^n[w_rw_l\br + |w_r|^2\bl].
\ee
Similarly, for $\xi^\dag=\bl^T$,
\be
( w\bl^T)_\k w = \ri Im(w_l)w_l\bl + \sum_{r>l}^n[w_rw_l\br - |w_r|^2\bl ].
\ee
We obtain then,
\be
\left\{
\ba
\XX_{Re(w_l}(w) &= 
\ri Re(w_l)w_l\bl + \ri\sum_{r>l}^n[w_rw_l\br + |w_r|^2\bl ]
+\lambda\ri\bl - \half(\ri w_l - \ri\overline{w_l})w\\
\XX_{Im(w_l)}(w) &= 
\ri Im(w_l)w_l\bl + \sum_{r>l}^n[w_rw_l\br - |w_r|^2\bl ]
- \lambda\bl - \half(w_l + \overline{w_l})w
\ea\right.
\ee
which combine to give
\be
\left\{\ba
\XX_{w_l}(w) &=
\ri w_l^2\bl + 2\ri\sum_{r>l}^nw_rw_l\br - \ri w_l w \\
\XX_{\overline{w}_l}(w) &= \ri |w_l|^2 \bl + 2\ri\sum_{r>l}^n |w_r|^2\bl + 2\lambda\ri\bl + \ri\overline{w}_l w
\ea
\right.
\ee
or, equivalently,
\be
\ba
\{w_k,w_l\} &= 2\ri \delta_{k>l}w_kw_l - \ri w_kw_l + \ri\delta_{kl}w_k^2\\ 
&=
\ri\sgn(k-l) w_kw_l,\\
\{w_k , \overline{w}_l\} &=
\ri\delta_{kl}|w_l|^2 + 2\ri\delta_{kl}\sum_{i>k} |w_i|^2  + 2\ri\lambda\delta_{kl} + \ri w_k\overline{w}_l \\
&=
\ri\delta_{kl}\left(2\sum_{r\geq k}|w_r|^2  + 2\lambda - |w_k|^2\right) + \ri w_k\overline{w}_l.
\ea
\ee
Using these, we get
\[
\{w_k,|w_l|^2\} = \ri\Bigl[ \bigl(1+\sgn(k-l)\bigr)|w_l|^2  + \delta_{kl}\bigl(2\lambda + 2\sum_{r\geq k} |w_r|^2 - |w_k|^2\bigr)\Bigr]w_k,
\]
and hence
\be
\{|w_k|^2,|w_l|^2\}=0.
\ee
Next, setting $w_k=|w_k|e^{\ri\phi_k}$, we have
\[
\ba
|w_k| \{e^{\ri\phi_k},|w_l|^2\} &= \{w_k,|w_l|^2\}\\ 
&=
\ri\Bigl[ \bigl(1- \delta_{kl} +\sgn(k-l)\bigr)|w_l|^2 + 2\delta_{kl}\bigl(\lambda+\sum_{r\geq k} |w_r|^2\bigr)\Bigr]\,|w_k|e^{\ri\phi_k},\\
\ea
\]
and then, from
\[
|w_k|e^{\ri\phi_l}\{e^{\ri\phi_k},|w_l|\} + |w_l|e^{\ri\phi_k}\{|w_k|,e^{\ri\phi_l}\} + |w_k|\,|w_l|\{e^{\ri\phi_k},e^{\ri\phi_l}\} = \ri\sgn(k-l)|w_k|\,|w_l|e^{\ri\phi_k}e^{\ri\phi_l},
\]
we obtain
\be
\{e^{\ri\phi_k},e^{\ri\phi_l}\}=0.
\ee
It is convenient to make the invertible change of variables, defined by
\be
G_k= \lambda+\sum_{r=k}^n|w_r|^2 ,
\ee
in terms of which the Poisson bracket relations are
\be
\ba
\{G_k,G_l\}&=0,\\
\{e^{\ri\phi_k},e^{\ri\phi_l}\}&=0,\\
\{e^{\ri\phi_k},G_l\} &= 2\ri\delta_{k\geq l}e^{\ri\phi_k}G_l.
\ea
\ee
The change of variables
\be
(e^{\ri\phi_1}\dots,e^{\ri\phi_n})\mapsto (e^{\ri\phi_1}, e^{\ri[\phi_2-\phi_1]},\dots,e^{\ri[\phi_n-\phi_{n-1}]})
\ee
diagonalizes the Poisson relations:
\be
\ba
\{e^{\ri\phi_1},G_l\} &= 2\ri\delta_{l1}e^{\ri\phi_1}G_1,\\
\{e^{\ri[\phi_k-\phi_{k-1}]},G_l\} &= 2\ri\delta_{kl} e^{\ri[\phi_k-\phi_{k-1}]}G_l,\quad k>1,
\ea
\ee
and so we obtain the corresponding symplectic form:
\[
\ba
\Omega_\lambda (w) &= \frac1{2\ri e^{\ri\phi_1}G_1}d(e^{\ri\phi_1})\wedge dG_1 
+ \frac1{2\ri}\sum_{k=2}^n\left( \frac1{e^{\ri[\phi_k-\phi_{k-1}]}}\frac1{G_k} d\left(e^{\ri[\phi_k-\phi_{k-1}]}\right)\wedge d G_k \right)\\
&=
\frac1{2\ri }\frac{d(e^{\ri\phi_1})}{e^{\ri\phi_1}}\wedge \frac{dG_1}{G_1} 
+ \frac1{2\ri}\sum_{k=2}^n\left( \frac{d(e^{\ri\phi_k})}{e^{\phi_k}} - \frac{d(e^{\ri\phi_{k-1}})}{e^{\ri\phi_{k-1}}}\right) \wedge \frac{d G_k}{G_k} \\ 
&=
\frac1{2\ri }\frac{d(e^{\ri\phi_n})}{e^{\ri\phi_n}}\wedge \frac{dG_n}{G_n} +
\frac1{2\ri}\sum_{k=1}^{n-1}\frac{d(e^{\ri\phi_k})}{e^{\ri\phi_k}}\wedge\left(\frac{dG_k}{G_k} - \frac{dG_{k+1}}{G_{k+1}}\right) .
\ea
\]
Using the identity $G_k-G_{k+1} = |w_k|^2$, the last bracketed term simplifies:
\[
\ba
\frac{dG_k}{G_k} - \frac{dG_{k+1}}{G_{k+1}} &= \frac1{G_kG_{k+1}} \Bigl(G_{k+1} d(G_{k+1} + |w_k|^2 ) - (G_{k+1}+|w_k|^2)dG_{k+1}\Bigr)\\
&=
\frac1{G_kG_{k+1}} \Bigl(G_{k+1} d( |w_k|^2 ) - |w_k|^2 dG_{k+1}\Bigr),
\ea
\]
and, writing the other derivatives in the form
\[
\ba
\frac{d(e^{\ri\phi_k})}{e^{\ri\phi_k}} &= \frac1{2|w_k|^2}(\overline{w}_kdw_k - w_kd\overline{w}_k),\\
d(|w_k|^2)&= \overline{w}_kdw_k + w_kd\overline{w}_k,
\ea
\]
we obtain
\be
\ba
\Omega_\lambda(w) &= \frac1{2\ri}\sum_{k=1}^n\frac12\frac{\overline{w}_kdw_k - w_k d \overline{w}_k }{|w_k|^2} \wedge \frac{\overline{w}_kdw_k + w_kd\overline{w}_k}{G_k}
-
\frac1{2\ri}\sum_{k=1}^{n-1} \frac12( \overline{w}_kdw_k - w_k d \overline{w}_k ) \wedge \frac{dG_{k+1}}{G_kG_{k+1}} \\
&=
-\frac\ri2\sum_{k=1}^n\frac1{G_k}dw_k\wedge d\overline{w}_k 
-
\frac\ri4\sum_{k=1}^{n-1}\frac1{G_kG_{k+1}}dG_{k+1}\wedge(\overline{w}_kdw_k - w_kd\overline{w}_k).
\ea
\ee
In this form it is easy to recognise that $\Omega_\lambda$ is a deformation of the canonical sympletic form on $\CC^n$: for $\lambda=1$, with $w=N^{-1}z$, we get $N^2\Omega_1(w)\to -\ri/2\sum_{k=1}^n dz_k\wedge d\overline{z}_k$ as $N\to\infty$.





\begin{thebibliography}{99}

\addcontentsline{toc}{section}{References}

    \setlength{\parskip}{0em}

 \bibitem{AM}
 A. Yu.~Alekseev and  A. Z.~Malkin,
 {\it Symplectic structures associated to Lie--Poisson groups},
 Comm. Math. Phys. {\bf 162} (1994) 147-173
 
\bibitem{AO}
G.~Arutyunov and E.~Olivucci,
{\it Hyperbolic spin Ruijsenaars--Schneider model from Poisson reduction},
{\tt  arXiv:1906.02619}


\bibitem{ES}
P.~Etingof and O.~Schiffmann,
{\it Lectures on quantum groups},
2nd edition, international press, 2002


\bibitem{FFM}
 M.~Fairon, L.~Feher and I.~Marshall, 
 {\it Trigonometric real form of the spin RS model of Krichever and Zabrodin},
{\tt arXiv:2007.08388}


\bibitem{GJS}
I.~Ganev, D.~Jordan and P.~Safronov, 
{\it The quantum Frobenius for character varieties and multiplicative quiver varieties},
{\tt arXiv:1901.11450v4}



\bibitem{H}
T. J.~Hodges: {\it On the Cremmer--Gervais quantization of $\SL(n)$} 
IMRN {\bf 10} (1995) 465-481


\bibitem{Lu}
J.-H.~Lu,
{\it Momentum mappings and reduction of Poisson actions},
 pp. 209-226; in: Symplectic Geometry, Groupoids, and Integrable Systems, Springer, 1991
 

\bibitem{LM}
 J.-H.~Lu and V.~Mouquin,
 {\it Mixed product Poisson structures associated to Poisson Lie groups and Lie bialgebras}, IMRN {\bf 19} (2017) 5919-5976


\bibitem{LW}
 J.-H.~Lu and A.~Weinstein, 
 {\it Poisson Lie groups, dressing transformations and Bruhat decompositions}, 
 J. Diff. Geom. {\bf 31} (1990) 501-526


\bibitem{LuYak}
 J.-H.~Lu and M.~Yakimov,
 {\it Group orbits and regular partitions of Poisson manifolds}
 Comm. math. phys. {\bf 283} (2008) 729-748


\bibitem{STS1}
  M. A.~Semenov-Tian-Shansky,
{\it Dressing transformations and Poisson group actions},
Publ. RIMS {\bf 21} (1985) 1237-1260


\bibitem{STS2}
  M.A.~Semenov-Tian-Shansky,
{\it Integrable systems: the r-matrix approach}, 
Publ. RIMS {\bf 1650} (2008) 


\bibitem{Yred}
M.~Yakimov, {\it Symplectic leaves of complex reductive Poisson Lie groups} Duke Math. J. {\bf 112} (2002) 453-509


\bibitem{Z}
S.~Zakrzewski,
{\it Phase spaces related to standard classical $r$-matrices},
J. Phys. A: Math. Gen. {\bf 29} (1996) 1841-1857;


\end{thebibliography}
\end{document}